\documentclass[11pt,onecolumn,draftclsnofoot]{IEEEtran}

\usepackage[cmex10]{amsmath}
\usepackage{amssymb}
\usepackage{epsfig,epsf,pgf}
\usepackage{cite}

\usepackage{enumerate}

\usepackage{amsthm}
\usepackage{psfrag}

\usepackage{times}
\usepackage{epsfig}
\usepackage{amsmath}
\usepackage{amsfonts}
\usepackage{algorithm}
\usepackage{algorithmic}
\usepackage{graphicx}
\usepackage{amssymb}
\usepackage{amstext}
\usepackage{latexsym}
\usepackage{color}
\usepackage{ifthen}
\usepackage{multirow}
\usepackage{verbatim}
\usepackage{array,tabularx}
\usepackage{epsf,pgf,pst-node,pstricks}
\usepackage{color}

\newcommand{\mbf}{\mathbf}
\newcommand{\mcl}{\mathcal}

\newcommand{\und}{\underline}

\newtheorem{theorem}{Theorem}
\newtheorem{lemma}{Lemma}

\newtheorem{claim}{Claim}

\theoremstyle{definition}
\newtheorem{definition}{Definition}
\newtheorem{example}{Example}
\newtheorem{remark}{Remark}
\newtheorem{corollary}{Corollary}

\begin{document}

\title{Optimal Deterministic Polynomial-Time Data Exchange for Omniscience}
\author{
\IEEEauthorblockN{Nebojsa Milosavljevic, Sameer Pawar, Salim El Rouayheb, Michael Gastpar,\IEEEauthorrefmark{2}\thanks{\IEEEauthorrefmark{2}Also with the School of Computer and Communication Sciences, EPFL, Lausanne, Switzerland.}\\ and Kannan Ramchandran} \\
\IEEEauthorblockA{Department of Electrical Engineering and Computer Sciences \\
University of California, Berkeley \\
Email: \{nebojsa, spawar, salim, gastpar, kannanr\}@eecs.berkeley.edu}
\thanks{This research was funded by the NSF grants
(CCF-0964018, CCF-0830788), a DTRA grant (HDTRA1-09-1-0032), and in part by an
AFOSR grant (FA9550-09-1-0120).}
}

\maketitle

\begin{abstract}
We study the problem of constructing a deterministic polynomial time algorithm that achieves omniscience, in a
rate-optimal manner, among a set of users that are interested in a common file but each has only partial knowledge
about it as side-information. Assuming that the collective information among all the users is sufficient to allow
the reconstruction of the entire file, the goal is to minimize the (possibly weighted) amount of bits that these users need
to exchange over a noiseless public channel in order for all of them to learn the entire file. Using established connections to the  multi-terminal secrecy problem, our algorithm also implies  a polynomial-time method for constructing a maximum size secret shared key in the presence of an eavesdropper.

We consider the
following types of side-information settings: (i) side information in the form  of uncoded fragments/packets of the
file, where the users' side-information consists of subsets of the file;  (ii) side information in the  form of
linearly correlated packets, where the users have access to linear combinations of the file packets; and (iii) the
general setting where the the users' side-information has an arbitrary (i.i.d.) correlation structure. Building on
results from combinatorial optimization, we provide a polynomial-time algorithm (in the number of users) that, first
finds the optimal rate allocations among these users,  then determines an explicit transmission scheme (\emph{i.e.}, a
description of which user should transmit what information) for cases (i) and (ii).
\end{abstract}

\section{Introduction}\label{sec:intro}
In the recent years cellular systems  have witnessed significant improvements in terms of data rates and are nearly
approaching theoretical limits in terms of the physical layer spectral efficiency. At the same time the rapid growth in
the popularity of data-enabled mobile devices, such as smart phones and tablets, far beyond the early adoption stage,
and correspondingly the increasing  demand for more throughput are challenging our ability to meet this demand
even with the current highly efficient cellular systems. One of the major bottlenecks in scaling the throughput with the
increasing number of mobile devices is the last mile wireless link between the base station and the mobile devices -- a
resource that is shared among many users served within the cell. This motivates investigating new ways where cell phone
devices can possibly cooperate among themselves to get the desired data in a peer-to-peer fashion without solely relying on the base station.

\begin{figure}
\begin{center}
\psset{unit=0.40mm}
\begin{pspicture}(0,0)(170,90)
\small{
\psframe(60,10)(110,30)\rput(85,20){\small{Base Station}}
\rput(85,2){$\{w_1,w_2,w_3,w_4\}$}
\rput(-6,20){$\left\{
                \begin{array}{c}
                  w_2 \\
                  w_3 \\
                  w_4 \\
                \end{array}
              \right\}$}
\psframe(12,15)(37,25)
\rput(25,20){user $1$}
\psframe(133,15)(159,25)
\rput(146,20){user $3$}
\rput(177,20){$\left\{
                 \begin{array}{c}
                   w_1 \\
                   w_2 \\
                   w_4 \\
                 \end{array}
               \right\}
$}
\rput(85,80){$\left\{
                \begin{array}{c}
                  w_1 \\
                  w_3 \\
                \end{array}
              \right\}
$}
\psframe(72,52)(98,62)
\rput(85,57){user $2$}
\psline{->}(60,20)(37,20)
\psline{->}(110,20)(133,20)
\psline{->}(85,30)(85,52)
}
\end{pspicture}
\end{center}
\caption{An example of the data exchange problem. A base station has a file formed of four packets $w_1,\dots, w_4\in
\mathbb{F}_{q}$ and wants to  deliver it to three users over an unreliable wireless channel. The base station
stops transmitting once the users collectively have all the packets, but may individually have only subsets of the
packets. For instance, here the base station stops after user 1, user 2 and user 3 have respectively packets $\{w_2,
w_3, w_4\}, \{w_1, w_3\},$ and $\{w_1, w_2, w_4\}$, which can now be regarded as side information. The users can then
cooperate among themselves to recover their missing packets. Here, the 3 users can reconcile their file with the
following optimal scheme that minimizes the total amount of communicated bits: user~$1$ transmits packet $w_4$, user
$2$ transmits $w_1+w_3$, and user $3$ transmits $w_2$, where the addition is in the  field $\mathbb{F}_{q}$.}
\label{fig:model_raw}
\end{figure}
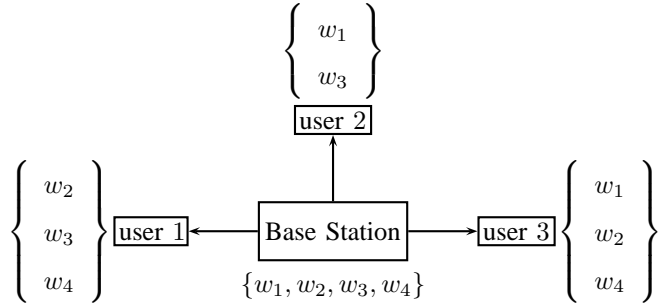

 An example of such a setting is shown in Figure~\ref{fig:model_raw}, where a base station
wants to deliver the same file to multiple  geographically-close users over an unreliable wireless downlink.  Such scenario may occur for instance when co-workers are using their tablets to  share and update files stored in the cloud (e.g., Dropbox), or when users, in the subway or a mall, are interested in watching the same popular video.
 For our example, let us suppose that the file
consists of four equally sized packets $w_1$, $w_2$, $w_3$ and $w_4$ belonging to some finite field $\mathbb{F}_q$.
Also, suppose that after few initial transmission attempts by the base station, the three users individually receive only
parts of the file (see Figure~\ref{fig:model_raw}), but collectively have the entire file. Now, if the mobile users
are in close vicinity  and can communicate with each other, then, it is much more desirable and efficient, in terms
of resource usage, to reconcile the file among users by letting them ``talk'' to each other without involving the  base station.  This cooperation has the following advantages:
\begin{itemize}
\item The connection to the base station is either unavailable after the initial phase of transmission, or it is too weak to meet the
delay requirement.
\item Transmissions within the close group of users is much more reliable than from any user to the base station due to geographical proximity.
\item Local communication among users has a smaller footprint in terms of interference, thus allowing one to use the shared resources (code, time or frequency)
freely without penalizing the base station's resources, \emph{i.e.}, higher resource reuse factor.
\end{itemize}

 The problem of reconciling a file among
multiple  wireless users having parts of it while minimizing  the  cost in terms of the total number of bits exchanged
is known in the literature as the {\em data exchange problem} and was introduced by El Rouayheb {\em et al.} in
\cite{SSS10}. In terms of the example considered here, if the $3$ users transmit $R_1,R_2$ and $R_3$ bits to reconcile
the entire file, the data exchange problem would correspond to minimizing the sum-rate $R_1+R_2+R_3$ such that,  when
the communication is over, all the users can recover the entire file. It can be shown here that the minimum sum-rate
required to reconcile the file is equal to 3 and can be achieved by the following coding scheme: user $1$
transmits packet $w_4$, user $2$ transmits $w_1+w_3$,  and user $3$ transmits $w_2$, where the addition is over the
underlying field $\mathbb{F}_{q}$. This corresponds to the optimal rate allocation $R_1=R_2=R_3=1$ symbol in
$\mathbb{F}_{q}$.

In a subsequent work, Sprinston {\em et al.} \cite{SSBE10} proposed a randomized algorithm
that with {\em high probability} achieves the minimum number of transmissions, given that the
field size $\mathbb{F}_q$ is large enough. Courtade \emph{et al.} \cite{CXW10} and  Tajbakhsh {\em et al.}
\cite{tajbakhshmodel} formulated this problem as a linear program (LP) and showed that the proposed LP under some
additional assumption\footnote{If users are allowed to split the packets into arbitrary number of smaller chunks.}, can
be solved in polynomial time. In a more general setting, one can consider  minimizing a different cost function, a
``weighted sum rate'', \emph{i.e.}, minimizing $\alpha_1R_1+\alpha_2R_2+\alpha_3R_3$, for some non-negative weights
$0 \leq \alpha_i<\infty$, $i=1,2,3$, to accommodate the scenario when transmissions from different users have different
costs. This problem was studied by Ozgul {\em et al.} \cite{ozgul2011algorithm}, where the authors proposed a
randomized algorithm that achieves this goal with {\em high probability} provided that the underlying field size is
large enough.

The results above consider  only the  simple form of the side-information where different users observe partial uncoded
``raw'' packets/fragments of the original file. Typically, content distribution networks use coding, such as Fountain
codes or linear network codes, to improve the system efficiency. In such  scenarios, the side-information representing
the partial knowledge gained by the users would be coded and in the form of linear combinations of the original file packets, rather than the raw packets themselves. The previous two cases of side information (``raw'' and
coded) can be regarded as special cases of the more general problem where the side-information has arbitrary
correlation among the observed data of different users and where the goal is to minimize the weighted total
communication (or exchange) cost to achieve omniscience. In \cite{CN04} Csisz\'ar and Narayan pose a related security
problem referred to as  the  ``multi-terminal key agreement" problem. They show  that achieving omniscience in minimum number of bits exchanged over the public channel is sufficient to maximize the size of the shared secret key.   This result establishes  the connection between the  ÔMulti-party key agreementÕ and the ÔData exchangeÕ problems.  The authors in \cite{CN04} solve the key agreement
problem by formulating it as a linear program (LP) with an exponential number of rate-constraints, corresponding to all
possible cut-sets that need to be satisfied, which has exponential complexity.
%


In this paper, we  make the following contributions. First,  we provide a \emph{deterministic polynomial time}
algorithm\footnote{The complexity of our proposed algorithm is $\mcl{O}(m^2 \cdot SFM(m))$, where $m$ is the number of
users and  $SFM(m)$ is the complexity of submodular function minimization.  To the  best of our knowledge, the fastest   algorithm  for SFM is given  by Orlin in \cite{orlin2009faster}, and has complexity $\mcl{O}(m^5\cdot \gamma+m^6)$, where $\gamma$ is complexity of
computing the submodular function.} for finding an optimal rate allocation, w.r.t. a linear weighted sum-rate cost, that
achieves omniscience among users with arbitrarily correlated side information. For the  data exchange problem,  this
algorithm computes the optimal rate allocation in polynomial time  for the case of  linearly coded side information
(including the ``raw'' packets case) and for the general linear cost functions (including the sum-rate case). Moreover,
for the ``multi-terminal key agreement" security problem of \cite{CN04}, this algorithm computes the secret key
capacity (maximum key length) in polynomial time. Second, for the the data exchange problem, with raw or linearly coded
side-information, we provide efficient methods for constructing linear network codes that can achieve  omniscience
among the users at the optimal rates with finite block lengths and zero-error.

The rest of the paper is organized as follows. In Section~\ref{sec:model}, we describe the model and formulate the
communication problem. Section~\ref{sec:comb} provides the necessary mathematical background  in combinatorial
optimization that will be needed for  constructing our algorithm. In Section~\ref{sec:CO_rates}, we describe the
polynomial time algorithm which finds an optimal rate allocation  that minimizes the sum-rate (non-weighted case). In
Section~\ref{sec:alpha}, we use the results of Section~\ref{sec:CO_rates} as a key building block to construct an
efficient algorithm for an arbitrary linear communication cost function. In Section~\ref{sec:fls}, we propose a
polynomial time code construction for the data exchange problem using results in network coding.
We conclude our work in Section~\ref{sec:conclusion}.

\section{System Model and Preliminaries}\label{sec:model}
In this paper, we consider a set up with $m$ user terminals that are interested in achieving omniscience of a
particular file or a random process. Let $X_1,X_2,\ldots,X_m$, $m \geq 2,$ denote the components of a discrete
memoryless multiple source (DMMS) with a given joint probability mass function. Each user terminal $i \in
\mcl{M}\triangleq \{1,2,\ldots,m\}$ observes $n$ i.i.d. realizations of the corresponding random variable $X_i$. The final goal
is for each terminal in the system to gain access to all other terminals' observations, \emph{i.e.}, to become
omniscient about the file or DMMS. In order to achieve this goal the terminals are allowed to communicate over a
noiseless public broadcast channel in multiple rounds and thus, may use interactive communication, meaning that
transmission by a user terminal at any particular time can be a function of its initial observations as well as the
past communication so far over the public broadcast channel. In \cite{CN04}, Csisz\'ar and Narayan showed that to
achieve the omniscience in a multi-terminal setup with general DMMS {\em interactive communication is not needed}. As a
result, in the sequel WLOG we can assume that the transmission of each terminal is only a function of its own initial
observations. Let $F_i := f_i(X^n_i)$ represent the transmission of the terminal $i \in \mcl{M}$, where $f_i(.)$ is any
desired mapping of the observations $X^n_i$. For each terminal to achieve omniscience, transmissions $F_i$, $i\in
\mcl{M}$, should satisfy,
\begin{align}
\lim_{n \rightarrow \infty} \frac{1}{n} H(X_{\mcl{M}}^n|\mbf{F},X_i^n) = 0,~~~\forall i \in \mcl{M}, \label{eq:decode}
\end{align}
where $X_{\mcl{M}}=(X_1,X_2,\ldots,X_m)$.
\begin{definition}
A rate tuple $\mbf{R}=(R_1,R_2,\ldots,R_m)$ is an {\em achievable communication for omniscience (CO) rate tuple} if there exists a communication scheme with transmitted messages $\mbf{F}=(F_1,F_2,\ldots,F_m)$ that satisfies \eqref{eq:decode},  {\em i.e.,} achieves omniscience, and is such that
\begin{align}
R_i = \lim_{n \rightarrow \infty} \frac{1}{n} H(F_i),~~~\forall i \in \mcl{M}.
\end{align}
\end{definition}


In the omniscience problem every terminal is a potential transmitter as well as a receiver. As a result, any set $\mcl{S}\subset \mcl{M}, \mcl{S} \neq \mcl{M},$ defines a cut corresponding to the partition between two sets $\mcl{S}$ and $\mcl{S}^c = \mcl{M} \setminus \mcl{S}$. It is easy to show using cut-set bounds that all the achievable CO rate tuple's necessarily belong to the following region
\begin{align}
\mcl{R}\triangleq \left\{\mbf{R}: R(\mcl{S})\geq H(X_{\mcl{S}}|X_{\mcl{S}^c}),~\mcl{S}\subset \mcl{M}\right\}, \label{cut_set}
\end{align}
where $R(\mcl{S}) = \sum_{i \in \mcl{S}} R_i$. Also, using a random coding argument,  it can be shown that the rate region $\mcl{R}$ is an achievable rate region \cite{CN04}. In \cite{pradhan03} and \cite{pradhan2005generalized} the authors provide explicit structured codes based on syndrome decoding that achieve the rate region for a Slepian-Wolf distributed source coding problem. This approach was further extended in \cite{stankovic06} to a multiterminal setting.

In this work, we aim to design a polynomial complexity algorithm that achieves omniscience among all the users while simultaneously minimizing an appropriately defined cost function over the rates. In the sequel we focus on the linear cost functions of the rates as an objective of the optimization problem. To that end, let $\und{\alpha} \triangleq (\alpha_1,\cdots,\alpha_m), 0 \leq \alpha < \infty$, be an $m-$dimensional vector of non-negative finite weights. We allow $\alpha_i$'s to be arbitrary non-negative constants, to account for the case when communication
of some group of terminals is more expensive compared to the others, \emph{e.g.}, setting $\alpha_1$ to be a large value compared to the other weights minimizes the rate allocated to the terminal $1$. This goal can be formulated as the following linear program which hereafter we denote by LP$_1(\und{\alpha})$:
\begin{align}
\min \sum_{i=1}^m \alpha_i R_i,~~~~\text{s.t.}~~~\mbf{R} \in \mcl{R}, \label{problem1}
\end{align}
We use $\mcl{R}(\und{\alpha})$ to denote the rate region of all minimizers of the above LP, and $R_{CO}(\und{\alpha})$ to denote the minimal cost.

\subsection*{Data Exchange Problem with linear correlation among users observations}
As mentioned in Section~\ref{sec:intro} efficient content distribution networks use coding such as fountain codes or linear network codes.
This results in users' observations to be in the form of linear combinations of the original packets forming the file, rather
than the raw packets themselves as is the case in conventional `Data Exchange problem'. This linear correlation source model is known in literature as {\em Finite linear source} \cite{CZ10}.

Next, we briefly describe the finite linear source model. Let $q$ be some power of a prime.
Consider the $N$-dimensional random vector $\mathbf{W} \in \mathbb{F}^N_{q^n}$
whose components are independent and uniformly distributed over the elements of $\mathbb{F}_{q^n}.$
Then, in the linear source model, the observation of $i^{th}$ user is simply given by
\begin{align}
\mathbf{X}_{i} = \mathbf{A}_i \mbf{W}, \ i \in \mcl{M}, \label{model:eq1}
\end{align}
where $\mathbf{A}_i \in \mathbb{F}_q^{\ell_i \times N}$ is an observation matrix\footnote{ The entries in the
observation matrix $A_i, \forall i \in {\cal M}$ denote the coefficients of the code, e.g., Fountain code or linear
network code, used by the base station and hence belong to the smaller field $\mathbb{F}_q$ rather than the field
$\mathbb{F}_{q^n}$ to which the data packets belong. This assumption is justified since the coding coefficients are typically stored in the packet in an overhead of size
negligible compared to the packet length.
}
for the user $i$.

It is easy to verify that for the finite linear source model,
\begin{align}
\frac{H(X_i)}{\log q^n} = \text{rank}(\mbf{A}_i). \label{rank_entropy}
\end{align}
Henceforth for the finite linear source model we will use the entropy of the
observations and the rank of the observation matrix interchangeably.

For the sake of brevity we use the following notation
\begin{align}
\text{rank}\left\{\left[
                    \begin{array}{c}
                      \mbf{A} \\
                      \mbf{B} \\
                    \end{array}
                  \right]
  \right\} &\triangleq \text{rank}(\mbf{A},\mbf{B}), \\
\text{rank}(\mbf{A}|\mbf{B}) &\triangleq \text{rank}(\mbf{A},\mbf{B})-\text{rank}(\mbf{B}).
\end{align}

Similar to the general DMMS model, for the finite linear source model an omniscience achievable rate tuple necessarily belongs to
\begin{align}
\mcl{R}_{de}\triangleq \left\{\mbf{R}: R(\mcl{S})\geq \text{rank}(\mbf{A}_{\mcl{S}}|\mbf{A}_{\mcl{S}^c}),~\mcl{S}\subset \mcl{M}\right\}, \label{cut_set_datexc}
\end{align}
where $R(\mcl{S}) = \sum_{i \in \mcl{S}} R_i$, and $\mbf{A}_{\mcl{S}}$ is a matrix obtained by stacking $\mbf{A}_{i}, \forall i \in \mcl{S}$. 
The rate $R_i$, $i\in \mcl{M}$ is the number of symbols in $\mathbb{F}_{q^n}$ user $i$ transmits
over the noiseless broadcast channel.
 \label{sec:model}
\section{Optimization over polyhedrons and Edmond's algorithm}

In this section we review results and techniques from the theory of  combinatorial optimization. These results  will form a key ingredient in finding a polynomial time algorithm for solving the rate minimization problem LP$_1(\und{\alpha})$ which will be described in Sections~\ref{sec:CO_rates} and \ref{sec:alpha}. The idea is to recast the  underlying rate region $\mcl{R}$,
defined by the cut-set constraints in \eqref{cut_set}, as a  polyhedron of some set function whose dual is \emph{intersecting submodular} which can be optimized in polynomial time. Then, we identify conditions under which
the optimization problem over the dual polyhedron and the original problem
have the same optimal solution.

Here, we state
the  definitions, theorems and algorithms that will be needed in the next sections. For a comprehensive exposition of combinatorial optimization,  we refer the interested reader to references \cite{schrijver2003combinatorial, F05}.

\begin{definition}[Polyhedron]
Let $f$ be a real function defined over the set $\mcl{M}=\{1,2,\ldots,m\}$, \emph{i.e.}, $f:2^{\mcl{M}}\rightarrow \mathbb{R}$
such that $f(\emptyset)=0$, where $2^{\mcl{M}}$ is the power set of $\mcl{M}$.
Let us define the \emph{polyhedron} $P(f,\leq)$ and the \emph{base polyhedron} $B(f,\leq)$ of $f$ as follows.
\begin{align}
P(f,\leq) & \triangleq \{\mbf{Z}~|~\mbf{Z} \in \mathbb{R}^m,~~\forall \mcl{S}\subseteq \mcl{M} : Z(\mcl{S})\leq f(\mcl{S}) \}, \label{f:poyh} \\
B(f,\leq) & \triangleq \{\mbf{Z}~|~\mbf{Z} \in P(f,\leq),~~Z(\mcl{M})=f(\mcl{M})\} \label{co_f:base},
\end{align}
where  $Z(\mcl{S})=\sum_{i \in \mcl{S}} Z_i$.
\end{definition}
\begin{example}\label{comb_exp1}
Consider the function $f$ defined over set $\mcl{M}=\{1,2\}$ such that $f(\emptyset)=0$, $f(\{1\})=4$, $f(\{2\})=3$, and
$f(\{1,2\})=6$. The polyhedron $P(f)$ is defined by the region $Z_1\leq 4$, $Z_2\leq 3$, and $Z_1+Z_2\leq 6$ (see Figure \ref{fig:polyh}).
For the base polyhedron there is the additional constraint $Z_1+Z_2=6$.
\begin{figure}[h]
\begin{center}
\includegraphics[scale=0.55]{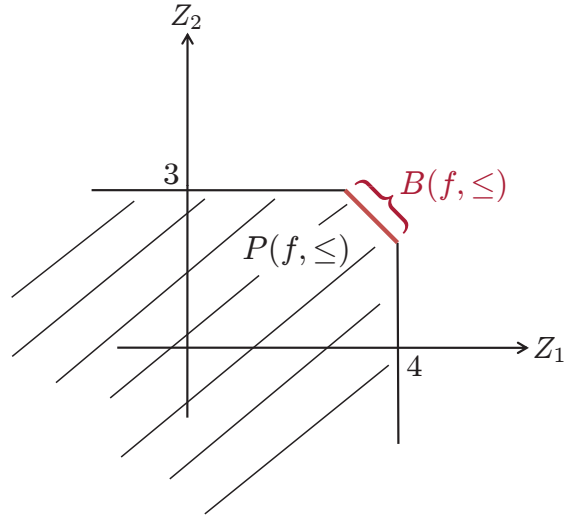}
\end{center}
\vspace{-0.2in}
\caption{Polyhedron $P(f,\leq)$ and the base polyhedron $B(f,\leq)$ for the function $f$ specified in Example \ref{comb_exp1}.}\label{fig:polyh}
\end{figure}
\end{example}
Notice that the base polyhedron $B(f,\leq)$ can be an empty set of vectors in general. For instance, if function $f$  in Example~\ref{comb_exp1}
is such that $f(\{1,2\})=8$ instead of $6$.
\begin{definition}[Dual function] \label{def:dual}
For a set function $f$ let us define its \emph{dual function} $f^{\star}:2^{\mcl{M}}\rightarrow \mathbb{R}$ as follows
\begin{align}
f^{\star}(\mcl{S}^c)=f(\mcl{M})-f(\mcl{S}),~~~\forall \mcl{S}\subseteq \mcl{M},
\end{align}
where $\mcl{S}^c=\mcl{M}\setminus \mcl{S}$.
\end{definition}
With the dual function $f^{\star}$, we associate its polyhedron and base polyhedron as follows
\begin{align}
P(f^{\star},\geq) & \triangleq \{\mbf{R}~|~\mbf{R} \in \mathbb{R}^m,~~\forall \mcl{S}\subseteq \mcl{M} : R(\mcl{S})\geq f^{\star}(\mcl{S}) \}, \\
B(f^{\star},\geq) & \triangleq \{\mbf{R}~|~\mbf{R} \in P(f^{\star},\geq),~~R(\mcl{M})=f^{\star}(\mcl{M})\}, \label{f_star_base}
\end{align}
\begin{lemma} \label{lm:app_opt1}
If $B(f,\leq)\neq \emptyset$ then, $B(f,\leq)=B(f^{\star},\geq)$ and $(f^{\star})^{\star}=f$.
\end{lemma}
Proof of Lemma \ref{lm:app_opt1} is provided in Appendix \ref{app:lm:app_opt1}.
For the set function $f$ from Example \ref{comb_exp1}, the polyhedron $P(f^{\star},\geq)$ and the base polyhedron $B(f^{\star},\geq)$ are
presented in Figure \ref{fig:polyh_dual}.
\begin{figure}[h]
\begin{center}
\includegraphics[scale=0.55]{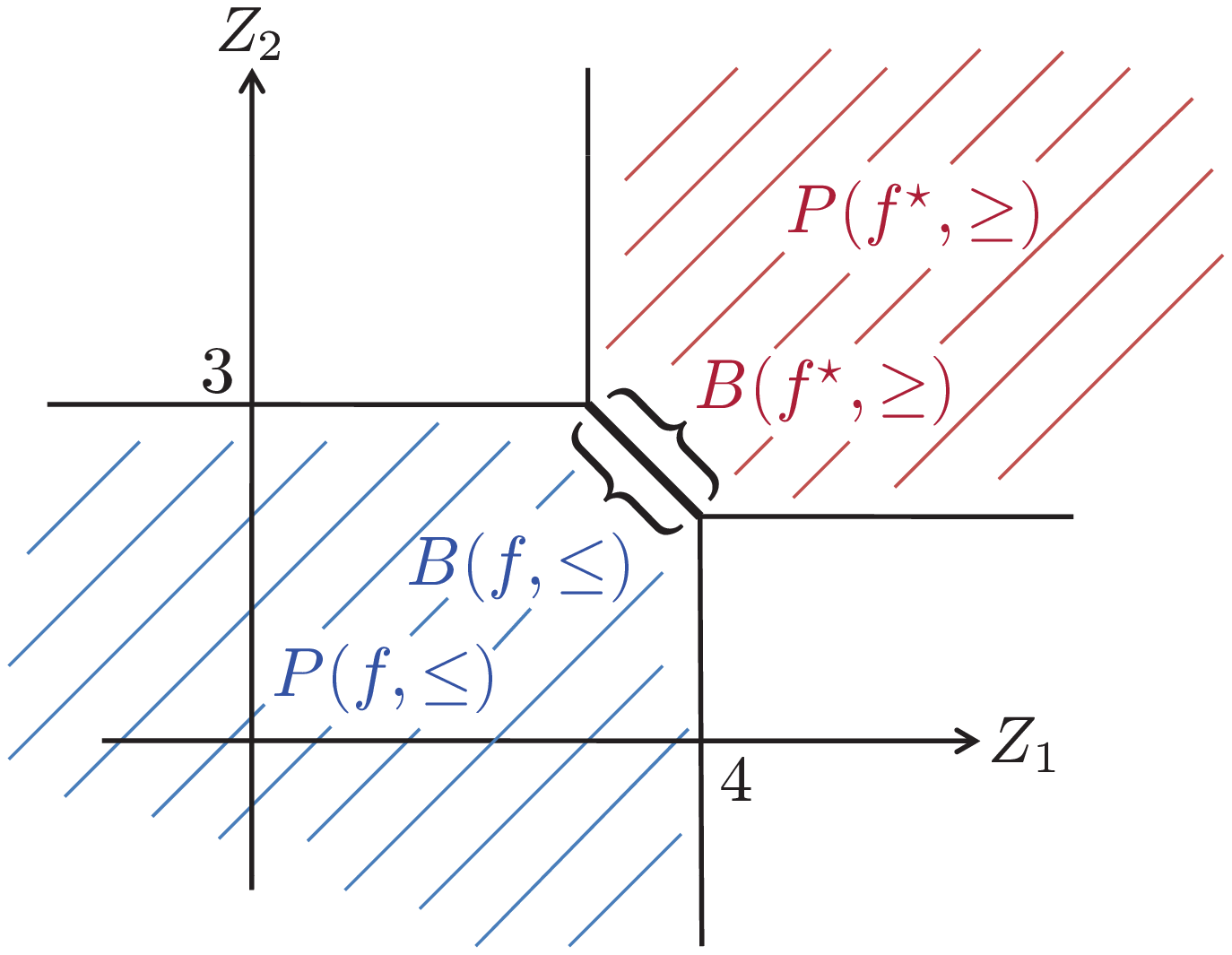}
\end{center}
\vspace{-0.2in}
\caption{Equivalence between $B(f,\leq)$ and $B(f^{\star},\geq)$ illustrated for the function $f$ provided in Example \ref{comb_exp1}.}\label{fig:polyh_dual}
\end{figure}
We say that two optimization problems are equivalent if they have the same optimal value and the same set of optimizers.
\begin{lemma} \label{app_rmk1}
If $B(f)\neq \emptyset$, then the following optimization problems are equivalent
\begin{align}
&\max_{\mbf{Z}} \sum_{i=1}^m Z_i,~~~\text{s.t.}~~\mbf{Z} \in P(f,\leq). \label{app_opt_lp1} \\
&\min_{\mbf{R}} \sum_{i=1}^m R_i,~~~\text{s.t.}~~\mbf{R} \in P(f^{\star},\geq). \label{app_opt_lp11}
\end{align}
\end{lemma}
Lemma \ref{app_rmk1} can be easily proved from the following argument provided in \cite{F05}.
Since $B(f,\leq)\neq \emptyset$, there exits a vector $\mbf{Z}$ such that $Z(\mcl{M})=f(\mcl{M})=f^{\star}(\mcl{M})$. Moreover,
$\mbf{Z}\in B(f,\leq)=B(f^{\star},\geq)$. Hence, $\mbf{Z}$ is a maximizer of the problem \eqref{app_opt_lp1} and a minimizer of the problem \eqref{app_opt_lp11}.

Next, we define the class of \emph{submodular functions} for which the maximization problem \eqref{app_opt_lp1}
has analytical solution.

\begin{definition}[Submodularity] \label{def:submodular}
A set function $f$ defined on the power set of $\mcl{M}$, $f:2^{\mcl{M}}\rightarrow \mathbb{R}$, where $f(\emptyset)=0$,
is called \emph{submodular} if
\begin{align}
f(\mcl{S})+f(\mcl{T})\geq f(\mcl{S}\cup \mcl{T})+f(\mcl{S}\cap \mcl{T}),
~~~\forall \mcl{S},\mcl{T} \subseteq \mcl{M}. \label{eq:submodular}
\end{align}
\end{definition}
\begin{remark} \label{rmk:base}
When $f$ is submodular, then $B(f,\leq)\neq \emptyset$.
\end{remark}
For a more general version of the problem \eqref{app_opt_lp1}
\begin{align}
\max_{\mbf{Z}} \sum_{i=1}^m \alpha_i Z_i,~~~\text{s.t.}~~\mbf{Z} \in P(f,\leq), \label{app_opt_lp2}
\end{align}
where $\alpha_i\geq 0$, for $i=1,\dots,m$, and $f$ is submodular, an analytical solution can be obtained using Edmond's algorithm.

\begin{theorem}[Edmond's greedy algorithm \cite{E70}]\label{thm:Edm}
When $f$ is  submodular, the maximization problem \eqref{app_opt_lp2} given by $\max_{\mbf{Z}} \sum_{i=1}^m Z_i$, s.t. $\mbf{Z}\in P(f)$,
can be solved analytically as follows.

\begin{align}
Z_{j(i)}&=f(\mcl{A}_i)-f(\mcl{A}_{i-1}), \ i=1,\dots, m, \nonumber
\end{align}
where $j(1),j(2),\ldots, j(m)$ is an ordering of $\{1,2,\ldots,m\}$ such that $\alpha_{j(1)}\geq \alpha_{j(2)}\geq \cdots \geq \alpha_{j(m)}$, and
\begin{align}
\mcl{A}_i&=\emptyset,~~~i=1, \nonumber \\
\mcl{A}_i&=\{j(1),j(2),\ldots,j(i)\},~~~i=2,3,\ldots,m. \nonumber
\end{align}
\end{theorem}
The following statement directly follows from Remark~\ref{rmk:base}.
\begin{remark}\label{rmk:edm_base}
When $f$ is submodular, a maximizer $\mbf{Z}$ of the optimization problem \eqref{app_opt_lp2} satisfies $\sum_{i=1}^m Z_i= f(\mcl{M})$.
\end{remark}
\begin{example}
In this example we illustrate Edmond's greedy algorithm by considering the set function $f$ from Example \ref{comb_exp1} and
the optimization problem
\begin{align}
\max_{\mbf{Z}} 5Z_1+Z_2,~~\text{s.t.}~\mbf{Z}\in P(f,\leq), \label{exp_opt}
\end{align}
where $\mbf{Z}=(Z_1,Z_2)$. Since $\alpha_1=5>\alpha_2=1$, we set $1$, $2$ to be the ordering of $\{1,2\}$, \emph{i.e.}, $j(1)=1$ and $j(2)=2$.
Then, by applying Edmond's algorithm
we obtain $Z_1=4$, $Z_2=2$ to be the maximizer of the problem \eqref{exp_opt}.
\begin{figure}[h]
\begin{center}
\includegraphics[scale=0.55]{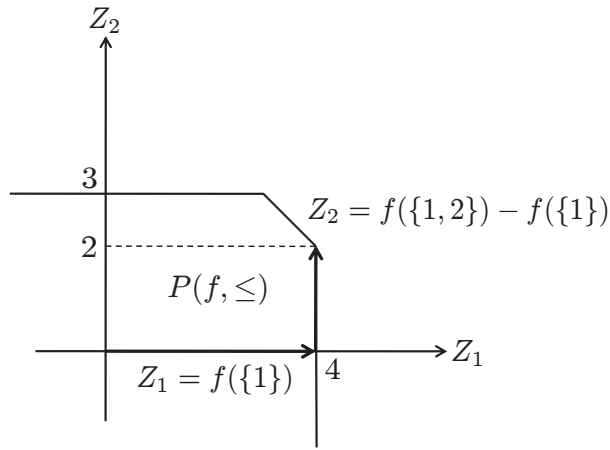}
\end{center}
\vspace{-0.2in}
\caption{Edmond's algorithm applied to the optimization problem \eqref{exp_opt}. Since $\alpha_1>\alpha_2$, the optimal ordering
        of $\{1,2\}$ is $1,2$.}\label{fig:edm}
\end{figure}
\end{example}
Edmond's algorithm is illustrated in Figure~\ref{fig:edm} for the case $\mcl{M}=\{1,2\}$. Notice that each iteration
of the algorithm reaches a boundary of the polyhedron $P(f,\leq)$ until it finally reaches the vertex of the base polyhedron $B(f,\leq)$.

In \cite{FT83}, it was shown that the following optimization problem can also be solved using Edmond's greedy algorithm.
\begin{corollary}\label{app_cor1}
When $f$ is submodular, then the optimization problem
\begin{align}
\min_{\mbf{R}} \sum_{i=1}^m \alpha_i R_i,~~~\text{s.t.}~~\mbf{R} \in B(f,\leq), \label{app_opt_lp3}
\end{align}
can be solved by using Edmond's algorithm where
$j(1),j(2),\ldots, j(m)$ is an ordering of $\mcl{M}$ such that $\alpha_{j(1)}\leq \alpha_{j(2)}\leq \cdots \leq \alpha_{j(m)}$.
\end{corollary}

Next, we introduce the class of \emph{intersecting submodular} functions which is instrumental to solving our communication for omniscience problem.

\begin{definition}[Intersecting Submodularity] \label{def:inter_submodular}
A function $f$ defined on the power set of $\mcl{M}$, $f:2^{\mcl{M}}\rightarrow \mathbb{R}$ is called an \emph{intersecting submodular} if
\begin{align}
f(\mcl{S})+f(\mcl{T})\geq f(\mcl{S}\cup \mcl{T})+f(\mcl{S}\cap \mcl{T}),
~~~\forall \mcl{S},\mcl{T}~\text{s.t.}~~\mcl{S}\cap \mcl{T}\neq \emptyset. \label{eq:inter_sub}
\end{align}
\end{definition}

Notice that every submodular function is also intersecting submodular. However, in general, Edmond's algorithm cannot be directly applied to solve the
maximization problem \eqref{app_opt_lp2} over the polyhedron of an intersecting submodular function.

In \cite{F05} it is shown that for every intersecting submodular function there exists a submodular
function such that both functions have the same polyhedron. This is formally stated in the following theorem.
\begin{theorem}[Dilworth truncation] \label{thm:dilworth}
For an intersecting submodular function $f:2^{\mcl{M}}\rightarrow \mathbb{R}$ with $f(\emptyset)=0$, there
exists a submodular function $g:2^{\mcl{M}}\rightarrow \mathbb{R}$ such that $g(\emptyset)=0$ and
$P(g,\leq)=P(f,\leq)$. The function $g$ can be expressed as
\begin{align}
g(\mcl{S})=\min_{\mcl{P}}\left\{\sum_{\mcl{V}\in \mcl{P}}f(\mcl{V}) : \text{$\mcl{P}$ is a partition of $\mcl{S}$}\right\}.  \label{dilw}
\end{align}
The function $g$ is called the \emph{Dilworth truncation} of $f$.
\end{theorem}
\begin{example} \label{exp_dilw}
Let $\mcl{M}=\{1,2\}$, and $f(\{1\})=4$, $f(\{2\})=3$, $f(\{1,2\})=8$. It is easy to verify that the function $f$
is intersecting submodular, but not fully submodular since $f(\{1\})+f(\{2\})<f(\{1,2\})$. Applying Dilworth truncation
to the function $f$, we obtain $g$, where $g(\{1\})=4$, $g(\{2\})=3$, $g(\{1,2\})=7$. Moreover, it can be checked that $P(g,\leq)=P(f,\leq)$.
\end{example}

If the  Dilworth
truncation $g$ of the intersecting submodular function $f$ is given, the optimization problem \eqref{app_opt_lp2} can be efficiently solved using Edmond's greedy algorithm.
However, finding the value of  function $g$, even for a single set $\mcl{S}\subseteq M$, involves a  minimization over a set of exponential
size (see \eqref{dilw}). This can be overcome using the facts that $P(g,\leq)=P(f,\leq)$, and that the maximizer of the problem \eqref{app_opt_lp2}
belongs to the base polyhedron $B(g,\leq)$ by Remark \ref{rmk:base}. The result is a  modified version of  Edmond's algorithm that can  solve the  optimization
problem in polynomial time.

\begin{lemma}[Modified Edmond's algorithm, \cite{F05}, \cite{NKI10}]\label{lm:modedm}
When $f$ is intersecting submodular, the maximization problem \eqref{app_opt_lp2} given by $\max_{\mbf{Z}} \sum_{i=1}^m Z_i$, s.t. $\mbf{Z}\in P(f)$,
can be solved as follows.
\begin{algorithm}
\caption{Modified Edmond's Algorithm}
\label{alg:modedm}
\begin{algorithmic}[1]
\STATE Set $j(1),j(2),\ldots, j(m)$ to be an ordering of $\{1,2,\ldots,m\}$ such that $\alpha_{j(1)}\geq \alpha_{j(2)}\geq \cdots \geq \alpha_{j(m)}$
\STATE Initialize $\mbf{Z}=\mbf{0}$.
\FOR {$i=1$ to $m$}
\STATE $Z_{j(i)}=\min_{\mcl{S}} \{f(\mcl{S})-Z(\mcl{S}) : j(i)\in \mcl{S},~~\mcl{S}\subseteq \mcl{A}_i \}$.
\ENDFOR
\end{algorithmic}
\end{algorithm}
\end{lemma}
The following statement directly follows from Theorem \ref{thm:dilworth} and Remark~\ref{rmk:edm_base}.
\begin{remark}\label{rmk:modedm_base}
When $f$ is intersecting submodular, a maximizer $\mbf{Z}$ of the optimization problem \eqref{app_opt_lp2} satisfies $\sum_{i=1}^m Z_i = g(\mcl{M})$,
where $g$ is the Dilworth truncation of $f$.
\end{remark}

What this algorithm essentially does is that at each iteration $i=1,2,\ldots,m$, it identifies a vector \\
$\left[
   \begin{array}{cccc}
     Z_{j(1)} & Z_{j(2)} & \ldots & Z_{j(i)} \\
   \end{array}
 \right]$ that lies on the boundary of the polyhedron $P(f,\leq)$. The polynomial complexity of the modified Edmond's algorithm
is due the fact that the function $f(\mcl{S})-Z(\mcl{S})$ is submodular since $\mcl{S}$ is not an empty set, and finding the minimum value of a submodular function  is known to polynomial (see \cite{orlin2009faster}).

\begin{example}
We illustrate the modified Edmond's algorithm for the function $f$ in Example \ref{exp_dilw}. Let us consider
maximization problem $\max_{\mbf{Z}} 5Z_1+Z_2$, s.t. $\mbf{Z}\in P(f,\leq)$. As mentioned above, at each iteration
of the algorithm, the optimal vector should lie on the boundary of the polyhedron $P(f,\leq)$. Hence, $Z_1=4$.
In the second iteration, in order to reach the boundary of $P(f,\leq)$, $Z_2$ can be either $f(\{1,2\})-Z_1=4$, or
$f(\{2\})=3$. Since the first choice results in the vector that does not belong to $P(f,\leq)$, the
solution is $Z_2=3$ (see Figure \ref{fig:modedm}).
\begin{figure}[h]
\begin{center}
\includegraphics[scale=0.55]{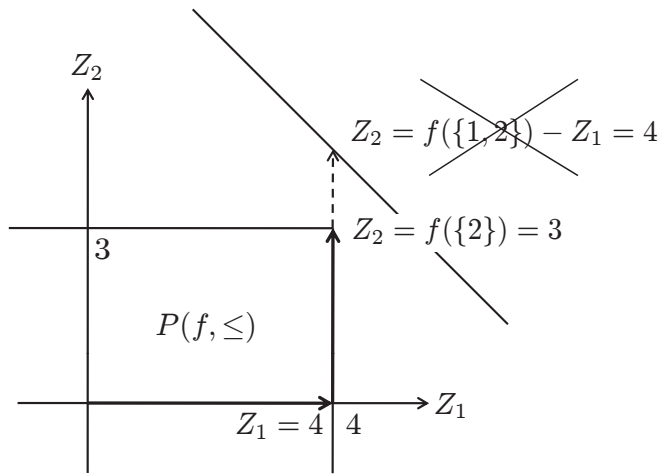}
\end{center}
\vspace{-0.2in}
\caption{Modified Edmond's algorithm applied to the maximization problem over the polyhedron $P(f,\leq)$,
         where $f(\emptyset)=0$, $f(\{1\})=4$, $f(\{2\})=3$, $f(\{1,2\})=8$.}\label{fig:modedm}
\end{figure}
\end{example}
\begin{theorem}[Complexity of the modified Edmond's algorithm \cite{NKI10}, \cite{F05}]\label{thm:ModEdm}
For an intersecting submodular function $f$, the optimization problem
\eqref{app_opt_lp2} can be solved in polynomial time using the modified version of Edmond's algorithm described in Lemma~\ref{lm:modedm}.
The complexity of this algorithm is $\mcl{O}(m \cdot SFM(m))$, where $SFM(m)$ is the complexity of minimizing
submodular function.
\end{theorem}
\begin{remark}
The submodular function minimization routine can be done in polynomial time. The best known algorithm to our knowledge
is proposed by Orlin in \cite{orlin2009faster}, and has complexity $\mcl{O}(m^5\cdot \gamma+m^6)$,
where $\gamma$ is complexity of computing the submodular function.
\end{remark}

 \label{sec:comb}
\section{Communication for Omniscience Rates}  \label{sec:CO_rates}

In this section we propose an efficient algorithm for computing a rate tuple which belongs to $\mcl{R}(\und{\alpha})$,
\emph{i.e.}, an optimal rate tuple w.r.t. the optimization problem
\begin{align}
\min_{\mbf{R}} \sum_{i=1}^m R_i,~~~\text{s.t.}~~\mbf{R} \in \mcl{R}.
\end{align}
We start with the special case when $\und{\alpha}=\left[
                                                    \begin{array}{cccc}
                                                      1 & 1 & \ldots & 1 \\
                                                    \end{array}
                                                  \right]$, henceforth denoted as $\und{\alpha}=\mbf{1}$.
This instance represents a key building block for solving the problem for general cost vector $\und{\alpha}$.
We begin by observing that the rate region defined in \eqref{cut_set} can be represented as a polyhedron of some set
function, say $f^{\star}$, to be defined later.
In this section we solve LP$_1(\mbf{1})$ by considering the dual set function $f$ of $f^{\star}$, and solving the corresponding
dual optimization problem. We show that it is possible to construct a function $f^{\star}$ defining the rate region $\mcl{R}$ such that its dual function $f$ is \emph{intersecting submodular}.
Therefore,  the underlying optimization problem can be solved  in polynomial time using the modified Edmond's algorithm.
Therefore, the optimization problem LP$_1(\mbf{1})$ can be stated as follows
\begin{align}
\min_{\mbf{R}} \sum_{i=1}^m R_i,~~~\text{s.t.}~~\mbf{R}\in P(f^{\star},\geq), \label{lp1_polyh}
\end{align}
where $P(f^{\star})$ is a polyhedron of a set function $f^{\star}$ such that $P(f^{\star},\geq)=\mcl{R}$.
To that end, we can choose
\begin{align}
f^{\star}(\mcl{S})=H(X_{\mcl{S}}|X_{\mcl{S}^c}),~~\forall \mcl{S} \subset \mcl{M}. \label{f_star}
\end{align}
Notice that the function $f^{\star}$ is not completely defined in \eqref{f_star} because the value
of $f^{\star}(\mcl{M})$ is missing. Therefore, we need to assign $f^{\star}(\mcl{M})$ such
that $P(f^{\star},\geq)=\mcl{R}$ and $B(f^{\star},\geq)\neq \emptyset$. The second condition
ensures equivalence between the optimization problem \eqref{lp1_polyh} and the corresponding dual problem (see Lemma \ref{app_rmk1}).
It is not hard to see that taking $f^{\star}(\mcl{M})=R_{CO}(\mbf{1})$ satisfies all the conditions above. Thus, we have
\begin{align}
f^{\star}(\mcl{S})=
\begin{cases}
H(X_{\mcl{S}}|X_{\mcl{S}^c}) & \text{if}~~\mcl{S} \subset \mcl{M}, \\
R_{CO}(\mbf{1})              & \text{if}~~\mcl{S}=\mcl{M}.
\end{cases}
\end{align}
Of course $R_{CO}(\mbf{1})$ is not known a priori, but this issue will be addressed later.
According to Definition~\ref{def:dual}, the dual set function $f$ of $f^{\star}$ has the following form
\begin{align}
f(\mcl{S})=
\begin{cases}
R_{CO}(\mbf{1})-H(X_{\mcl{S}^c}|X_{\mcl{S}}) & \text{if}~~\emptyset \neq \mcl{S} \subseteq \mcl{M}, \label{fcn:f_rco} \\
0                                            & \text{if}~~\mcl{S}=\emptyset.
\end{cases}
\end{align}
Using the duality result in Lemma \ref{app_rmk1}, it follows that the optimization problem \eqref{lp1_polyh} is equivalent to
\begin{align}
\max_{\mbf{Z}} \sum_{i=1}^m Z_i,~~~\text{s.t.}~~~\mbf{Z}\in P(f,\leq). \label{problem2}
\end{align}
To avoid cumbersome expressions, hereafter we use $P(f)$ and $B(f)$ to denote $P(f,\leq)$ and $B(f,\leq)$, respectively.
Hence, the optimal value of the optimization problem \eqref{problem2} is $R_{CO}(\mbf{1})$. However, the
value of $R_{CO}(\mbf{1})$ is not known a priori. To that end, let us replace $R_{CO}(\mbf{1})$ in \eqref{fcn:f_rco}
with a variable $\beta$, and construct a two-argument function $f(\mcl{S},\beta)$ as follows.
\begin{align}
f(\mcl{S},\beta)\triangleq
\begin{cases}
\beta-H(X_{\mcl{S}^c}|X_{\mcl{S}}) & \text{if}~~\emptyset \neq \mcl{S} \subseteq \mcl{M}, \label{fcn:f}\\
0                                            & \text{if}~~\mcl{S}=\emptyset.
\end{cases}
\end{align}
\begin{lemma} \label{lm:f}
Function $f(\mcl{S},\beta)$ defined in \eqref{fcn:f} is intersecting submodular. When $\beta \geq H(X_{\mcl{M}})$, the function
$f(\mcl{S},\beta)$ is submodular.
\end{lemma}
Proof of Lemma \ref{lm:f} is provided in Appendix \ref{app:lm:f}.
Considering the optimization problem
\begin{align}
\max_{\mbf{Z}} \sum_{i=1}^m Z_i,~~~\text{s.t.}~~~\mbf{Z}\in P(f,\beta), \label{mod_problem2}
\end{align}
as a function of $\beta$, the goal is to identify  its characteristics at the  point $\beta=R_{CO}(\mbf{1})$.
Hereafter, we refer to the optimization problem \eqref{mod_problem2} as LP$_2(\beta)$.
\begin{theorem} \label{thm:beta}
The optimal value $R_{CO}(\mbf{1})$ can be obtained as follows
\begin{align}
R_{CO}(\mbf{1}) = \min \beta~~\text{such that $\beta$ is the optimal value of LP$_2(\beta)$}. \label{opt:min_beta}
\end{align}
\end{theorem}
\begin{proof}
We prove this theorem by contradiction. First, notice that $\beta=R_{CO}(\mbf{1})$ is a feasible solution
for the optimization problem \eqref{opt:min_beta}.
Next, let us assume that for some $\beta'<R_{CO}(\mbf{1})$ there exists a vector $\mbf{Z}$ that is a maximizer
of the problem LP$_2(\beta')$ such that $Z(\mcl{M})=\beta'=f(\mcl{M},\beta')$.
Since $\mbf{Z} \in P(f,\beta')$ it must satisfy the following set of inequalities
\begin{align}
Z(\mcl{S})\leq \beta' - H(X_{\mcl{S}^c}|X_{\mcl{S}}),~~~\forall \emptyset \neq \mcl{S} \subseteq \mcl{M}. \label{thm:ineq1}
\end{align}
Since $\beta'=Z(\mcl{M})$, and $Z(\mcl{S}^c)=Z(\mcl{M})-Z(\mcl{S})$, we can write \eqref{thm:ineq1} as
\begin{align}
Z(\mcl{S}^c)\geq H(X_{\mcl{S}^c}|X_{\mcl{S}}),~~~\forall \emptyset \neq \mcl{S}\subset \mcl{M}.
\end{align}
Therefore, $\mbf{Z} \in \mcl{R}$ is a feasible rate tuple w.r.t. the optimization problem LP$_1(\mbf{1})$ and, hence,
it must hold that $\beta'\geq R_{CO}(\mbf{1})$.
This is in contradiction with our previous statement that $\beta'<R_{CO}(\mbf{1})$.
\end{proof}
Since $R_{CO}(\mbf{1})$ can be trivially upper bounded by $H(X_{\mcl{M}})$ and lower bounded by $0$, we can restrict
the search space in \eqref{opt:min_beta} to $0\leq \beta \leq H(X_{\mcl{M}})$.

Function $f(\mcl{S},\beta)$ is intersecting submodular for the case of interest when $0\leq \beta\leq H(X_{\mcl{M}})$.
As noted in Theorem \ref{thm:dilworth}, for the intersecting submodular function $f(\mcl{S},\beta)$, there exists a
submodular function, here denoted by Dilworth truncation $g(\mcl{S},\beta)$, such that $P(f,\beta)=P(g,\beta)$.
\begin{align}
g(\mcl{S},\beta)=\min_{\mcl{P}}\left\{\sum_{\mcl{V}\in \mcl{P}}f(\mcl{V},\beta) : \text{$\mcl{P}$ is a partition of $\mcl{S}$}\right\}.  \label{dilw1}
\end{align}
\begin{definition}
Let $\mcl{P}(\beta)$ denote an optimal partitioning of the set $\mcl{M}$ according to \eqref{dilw1} for the given $\beta$.
\end{definition}
From Remark \ref{rmk:modedm_base} it follows that $g(\mcl{M},\beta)$ is the optimal value of the optimization problem
LP$_2(\beta)$ for any given $\beta$. Hence, it can be obtained in polynomial time by applying the modified Edmond's algorithm to
the set function $f(\mcl{S},\beta)$. Moreover, the corresponding optimal partition $\mcl{P}(\beta)$ can be efficiently obtained by adding
two additional steps to  the modified Edmond's algorithm
as shown in \cite{NKI10} and \cite{F05} (see Algorithm~\ref{alg:partition} in  Appendix \ref{app:dilw}).

From Theorem~\ref{thm:beta}, it follows that the optimal omniscience rate $R_{CO}(\mbf{1})$ can be calculated as follows:
\begin{align}
R_{CO}(\mbf{1})=\min_{0\leq \beta \leq H(X_{\mcl{M}})} \beta,~~~\text{s.t.}~~g(\mcl{M},\beta)=\beta. \label{eq:rco}
\end{align}
Notice that $g(\mcl{M},\beta)=f(\mcl{M},\beta)=\beta$ whenever the optimal partitioning of the set $\mcl{M}$
according to \eqref{dilw1} is of cardinality $1$, \emph{i.e.}, $\mcl{P}(\beta)=\{\{\mcl{M}\}\}$.

In the further text we show how to
solve the optimization problem \eqref{eq:rco} with at most $m$ calls of the modified Edmond's algorithm, which makes the
complexity of the entire algorithm polynomial in $m$.
From \eqref{dilw1} it follows that for every $\beta$, the function $g(\mcl{M},\beta)$ can be represented as
\begin{align}
g(\mcl{M},\beta) = |\mcl{P}(\beta)|\beta - \sum_{\mcl{S}\in \mcl{P}(\beta)} H(X_{\mcl{S}^c}|X_{\mcl{S}}). \label{fcn:g_beta}
\end{align}
Therefore, $g(\mcl{M},\beta)$ is piecewise linear in $\beta$.
\begin{lemma} \label{lm:nondec}
Function $g(\mcl{M},\beta)$ has the following properties
\begin{enumerate}
\item It has at most $m$ linear segments.
\item It has non-increasing slope, \emph{i.e.}, $g(\mcl{M},\beta)$ is a concave function.
\item The last linear segment is of slope $1$.
\end{enumerate}
Moreover, $\beta=R_{CO}(\mbf{1})$ represents a breakpoint
in $g(\mcl{M},\beta)$ between the linear segment with slope $1$ and consecutive linear segment with the larger slope.
\end{lemma}
The proof of Lemma \ref{lm:nondec} is provided in Appendix \ref{app:lm:nondec}.
From \eqref{fcn:g_beta} it follows that the  slope of the function $g(\mcl{M},\beta)$ is equal to the cardinality of the optimal partition $\mcl{P}(\beta)$. Since there are at most $m$ linear segments in $g(\mcl{M},\beta)$, we can solve for the breakpoint of interest
according to Lemma \ref{lm:nondec} in polynomial time by performing a binary search. We explain this procedure
on a simple case described in Figure  \ref{fig:rco}.
\begin{figure}[h]
\begin{center}
\includegraphics[scale=0.6]{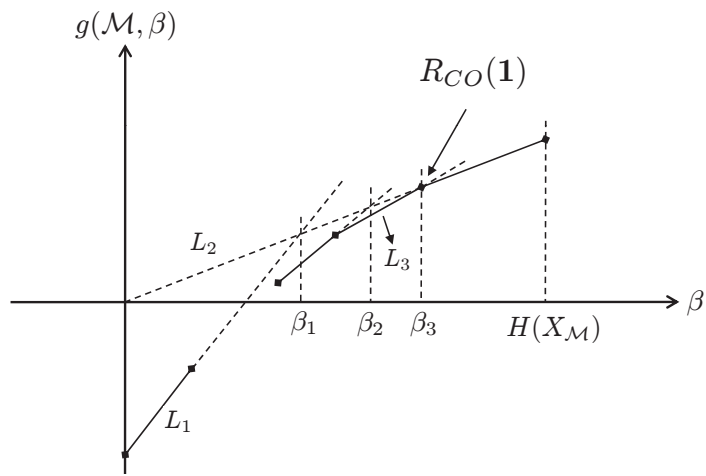}
\end{center}
\vspace{-0.2in}
\caption{Optimal $R_{CO}(\mbf{1})$ can be obtained by intersecting linear segments. First, we intersect the line $L_1$ which corresponds to $\beta=0$, with
         the $45$-degree line $L_2$. The intersecting point $\beta_1$ belongs to the linear segment with slope greater
         than $1$. Then, intersecting the segment $L_3$ to which $\beta_1$ belongs to with the $45$-degree line $L_2$, we obtain
         $\beta_2$, and finally $\beta_3$ after one more intersection. Since the linear segment at $\beta_3$ has
         slope $1$, we conclude that $\beta_3=R_{CO}(\mbf{1})$.} \label{fig:rco}
\end{figure}
From Lemma \ref{lm:nondec} we have that $\beta=R_{CO}(\mbf{1})$ is a breakpoint of $g(\mcl{M},\beta)$ between the linear segment
with slope $1$ and consecutive linear segment with the larger slope. Moreover, for every $\beta$ one can obtain a value of
$g(\mcl{M},\beta)$ and the corresponding optimal partition $\mcl{P}(\beta)$ w.r.t. \eqref{dilw1}
in polynomial time using Algorithm \ref{alg:partition} in Appendix \ref{app:dilw}.
Due to concavity of $g(\mcl{M},\beta)$, the following algorithm will converge to the breakpoint $\beta=R_{CO}(\mbf{1})$
in at most $m$ iterations.

Since $R_{CO}(\mbf{1})\geq 0$, we start by, first, intersecting the line $L_1$ which belongs to the linear segment when $\beta=0$ and the $45$-degree line $L_2$ which
corresponds to the last (rightmost) linear segment. Slope of the line $L_1$ as well as its value can be obtained in polynomial time
by applying Algorithm \ref{alg:partition} for $\beta=0$. Since the function $g(\mcl{M},\beta)$ is piecewise linear and concave, the point of
intersection $\beta_1$ must belong to the linear segment with slope smaller than $|\mcl{P}(0)|$, \emph{i.e.}, $|\mcl{P}(\beta_1)|<|\mcl{P}(0)|$.
$\beta_1$ can be obtained by equating $\beta$ with $\sum_{\mcl{S} \in \mcl{P}(0)} \beta - H(X_{\mcl{S}^c}|X_{\mcl{S}})$. Hence,
\begin{align}
\beta_1=\frac{\sum_{\mcl{S} \in \mcl{P}(0)}  H(X_{\mcl{S}^c}|X_{\mcl{S}})}{|\mcl{P}(0)|-1}. \label{eq:beta1}
\end{align}
Next, by applying Algorithm \ref{alg:partition} for $\beta=\beta_1$, we get $(g(\mcl{M},\beta_1),\mcl{P}(\beta_1))$.
Since $|\mcl{P}(\beta_1)|>1$ (see Figure \ref{fig:rco}), we have not reached the breakpoint of interest yet, because
$R_{CO}(\mbf{1})$ belongs to the linear segment of slope $1$.
Thus, we proceed by intersecting the line $L_3$ which belongs to the linear segment when $\beta=\beta_1$ with the
the $45$-degree line $L_2$. Like in the previous case, we obtain $\beta_2=\frac{\sum_{\mcl{S} \in \mcl{P}(\beta_1)}  H(X_{\mcl{S}^c}|X_{\mcl{S}})}{|\mcl{P}(\beta_1)|-1}$.
Since $|\mcl{P}(\beta_2)|>1$, we need to perform one more intersection to obtain $\beta_3$ for which $|\mcl{P}(\beta_3)|=1$. Hence, $\beta_3=R_{CO}(\mbf{1})$.
For an arbitrary $g(\mcl{M},\beta)$, the binary search algorithm can be constructed as follows.

\begin{algorithm}
\caption{Achieving a rate tuple from the region $\mcl{R}(\mbf{1})$}
\label{alg:bin_search}
\begin{algorithmic}[1]
\STATE Initialize $\beta=0$.
\WHILE {$|\mcl{P}(\beta)|>1$}
\STATE $\beta=\frac{\sum_{\mcl{S} \in \mcl{P}(\beta)}  H(X_{\mcl{S}^c}|X_{\mcl{S}})}{|\mcl{P}(\beta)|-1}$,
       where $\mcl{P}(\beta)$ is obtained from Algorithm \ref{alg:partition}.
\ENDWHILE
\STATE $\beta=R_{CO}(\mbf{1})$.
\end{algorithmic}
\end{algorithm}
It is not hard to see that Algorithm \ref{alg:bin_search} executes at most $m$ iterations, since with each iteration
the intersection point moves to the right to some other linear segment until it hits $R_{CO}(\mbf{1})$ (see Figure \ref{fig:rco}).

Therefore, Algorithm~\ref{alg:bin_search} calls Algorithm~\ref{alg:partition} at most $m$ times.
Since the complexity of Algorithm \ref{alg:partition} is $\mcl{O}(m\cdot SFM(m))$ (see Appendix~\ref{app:dilw}),
the total complexity of obtaining a rate tuple that belongs to $\mcl{R}(\mbf{1})$ through Algorithm~\ref{alg:bin_search} is $\mcl{O}(m^2\cdot SFM(m))$.

\section{Achieving a rate tuple that belongs to  $\mcl{R}(\und{\alpha})$} \label{sec:alpha}

In this section we investigate the problem of computing a rate tuple that belongs to $\mcl{R}(\und{\alpha})$, where
$0\leq \alpha_i<\infty$, $i=1,2,\ldots,m$. We propose an algorithm of polynomial complexity that is based on the
results we derived for the $\mcl{R}(\mbf{1})$ case.

Let us start with restating the optimization problem LP$_1(\und{\alpha})$ in the following way.
\begin{align}
\min_{\beta} \min_{\mbf{R}} \sum_{i=1}^m \alpha_i R_i~~~\text{s.t.}~~R(\mcl{M})=\beta,~~R(\mcl{S})\geq H(X_{\mcl{S}}|X_{\mcl{S}^c}),
~~\forall \mcl{S}\subset \mcl{M} \label{eq:lp3}
\end{align}
where $\beta \geq R_{CO}(\mbf{1})$.
Hereafter we denote optimization problem \eqref{eq:lp3} by LP$_3(\alpha)$.
This interpretation of the problem LP$_1(\und{\alpha})$ corresponds to finding its optimal value by searching over all
achievable sum rates $R(\mcl{M})$. Let us focus on the second term in optimization \eqref{eq:lp3}.
\begin{align}
\min_{\mbf{R}} \sum_{i=1}^m \alpha_i R_i~~~\text{s.t.}~~R(\mcl{M})=\beta,~~R(\mcl{S})\geq H(X_{\mcl{S}}|X_{\mcl{S}^c}),
~~\forall \mcl{S}\subset \mcl{M}. \label{eq:modlp3}
\end{align}
Observe that the rate region in \eqref{eq:modlp3} constitutes a base polyhedron $B(f^{\star},\beta,\geq)$, where
\begin{align}
f^{\star}(\mcl{S},\beta)=
\begin{cases}
H(X_{\mcl{S}}|X_{\mcl{S}^c})  &  \text{if}~\mcl{S}\subset \mcl{M}, \label{fcn:f_star} \\
\beta                         &  \text{if}~\mcl{S}=\mcl{M}.
\end{cases}
\end{align}
Since $\beta\geq R_{CO}(\mbf{1})$ we have that $B(f^{\star},\beta,\geq)\neq \emptyset$. From Lemma \ref{lm:app_opt1}
it follows that $B(f^{\star},\beta,\geq)=B(f,\beta)$, where $f(\mcl{S},\beta)$, defined in \eqref{fcn:f}, is a dual set function
of $f^{\star}(\mcl{S},\beta)$. Hence, the optimization problem \eqref{eq:modlp3} is equivalent to
\begin{align}
\min_{\mbf{R}} \sum_{i=1}^m \alpha_i R_i~~~\text{s.t.}~~\mbf{R}\in B(f,\beta). \label{lp3:mod}
\end{align}
In Corollary \ref{app_cor1} we implied that for any fixed $\beta\geq R_{CO}(\mbf{1})$ the optimization problem \eqref{lp3:mod} can be solved using Edmond's algorithm,
with $j(1),j(2),\ldots,j(m)$ being the ordering of $\mcl{M}$ such that $\alpha_{j(1)}\leq \alpha_{j(2)}\leq \cdots \leq \alpha_{j(m)}$.
However, since the function $f(\mcl{S},\beta)$ is intersecting submodular, it is necessary to apply the modified
version of Edmond's algorithm provided in Lemma \ref{lm:modedm} to obtain an optimal rate tuple w.r.t. \eqref{lp3:mod}.

Let $h(\beta)$ denote the optimal value of the optimization problem defined in \eqref{lp3:mod}
\begin{align}
h(\beta) = \min_{\mbf{R}} \sum_{i=1}^m \alpha_i R_i~~~\text{s.t.}~~\mbf{R}\in B(f,\beta). \label{fcn:h}
\end{align}
To that end, we can state problem LP$_3(\und{\alpha})$ as
\begin{align}
\min_{\beta} h(\beta),~~~\text{s.t.}~~\beta \geq R_{CO}(\mbf{1}).
\end{align}
With every $\beta\geq R_{CO}(\mbf{1})$ we associate an optimal rate vector $\mbf{R}$ w.r.t. optimization problem \eqref{fcn:h}.
Next, we show some basic properties of the function $h(\beta)$.

\begin{lemma}\label{lm:h_prop}
Function $h(\beta)$ defined in \eqref{fcn:h} is continuous and convex when $\beta \geq R_{CO}(\mbf{1})$.
\end{lemma}
Proof of Lemma~\ref{lm:h_prop} is provided in Appendix~\ref{app:lmh}.

\subsection{Gradient Descent Method} \label{sec:gdm}

From Lemma~\ref{lm:h_prop} it immediately follows that we can apply a gradient
descent algorithm to minimize the function $h(\beta)$. However, in order to do that, at every point $\beta$, we need
to know the value of $h(\beta)$ as well as its derivative. As mentioned above, an optimal rate tuple
that corresponds to the function $h(\beta)$ can be obtained by applying the modified Edmond's algorithm
to the problem \eqref{lp3:mod}.
From Lemma \ref{lm:modedm} it follows that the optimal rate vector with respect to the optimization
problem \eqref{lp3:mod}, has the following form.
\begin{align}
R_i = b_i\cdot \beta + c_i,~~\forall i\in \mcl{M}, \label{eq:represent}
\end{align}
where $b_i\in \mathbb{Z}$, and $c_i$ is a constant which corresponds to a summation of some conditional entropy terms.
Moreover, it follows that the coefficients $(b_i,c_i)$, $i=1,2,\ldots,m$, depend only on the value of $\beta$ (they do not
depend on the weight vector $\und{\alpha}$).
\begin{lemma} \label{lm:derivative}
Function $h(\beta)$ is piecewise linear in $\beta$. For a fixed $\beta \geq R_{CO}(\mbf{1})$
the values of  $h(\beta)$ and $\frac{d h(\beta)}{d \beta}$ can be obtained in $\mcl{O}(m \cdot SFM(m))$
time by applying the modified Edmond's algorithm to the ordering of $\mcl{M}$ specified in Corollary \ref{app_cor1}.
Derivative of $h(\beta)$ can be calculated by expressing the optimal rates $R_i$, $i \in \mcl{M}$, as $R_i = b_i \cdot \beta +c_i$
in each iteration of the modified Edmond's algorithm. Then,
\begin{align}
\frac{d h(\beta)}{d \beta}=\sum_{i=1}^m \alpha_i \cdot b_i.
\end{align}
\end{lemma}
To make the gradient descent algorithm more efficient, it is useful to make a search space as tight as possible.
So far, we showed that the minimizer of the problem LP$_3(\und{\alpha})$ belongs to the region $[R_{CO}(\mbf{1}),\infty)$.
Combining the results of Lemma \ref{lm:h_prop} and Lemma \ref{lm:derivative}, we have the following bound.
\begin{lemma}
Let $\beta^{\star}$ be the minimizer of the optimization problem LP$_3(\und{\alpha})$. Then,
\begin{align}
R_{CO}(\mbf{1})\leq \beta^{\star}\leq H(X_{\mcl{M}}).
\end{align}
\end{lemma}
\begin{proof}
Note that the function $f(\mcl{S},\beta)$  is submodular when $\beta=H(X_{\mcl{M}})$ (see Lemma \ref{lm:f}).
Optimization problem \eqref{eq:modlp3} for $\beta=H(X_{\mcl{M}})$, can be solved by applying Edmond's algorithm (see Theorem \ref{thm:Edm})
to the optimization problem \eqref{lp3:mod}.
It is easy to verify that the optimal rates have the following form:
\begin{align}
R_{j(1)}&=\beta + c_{j(1)}, \nonumber \\
R_{j(i)}&=c_{j(i)},~~i\in \{2,3,\ldots,m\}. \nonumber
\end{align}
Hence,
\begin{align}
h(\beta=H(X_{\mcl{M}})) = \alpha_{j(1)}\beta + \sum_{i=1}^m \alpha_{j(i)} c_{j(i)}. \nonumber
\end{align}
Since $\alpha_{j(1)}\geq 0$, and function $h(\beta)$ is convex, it immediately follows
that $\beta^{\star}\leq H(X_{\mcl{M}})$.
\end{proof}
Since the function $h(\beta)$ is continuous and differentiable, we can find its minimum, and therefore solve the optimization problem LP$_3(\und{\alpha})$,
by applying a gradient descent algorithm. However,
in general case, we can only reach the optimal point up to some precession $\varepsilon$.
In order to be at most $\varepsilon$ away from the optimal solution, the gradient descent method executes
approximately  $\mcl{O}(\log \frac{1}{\varepsilon})$ iterations \cite{boyd2004convex}. Therefore, the total complexity of
obtaining a rate tuple with a sum rate that is at most $\varepsilon$ away from the optimal one   is $\mcl{O}(m^2 \cdot SFM(m)+ \log \frac{H(X_{\mcl{M}})}{\varepsilon} \cdot m \cdot SFM(m))$, where the first term corresponds to the complexity of finding $R_{CO}(\mbf{1})$.

Before we go any further, let us briefly analyze a solution to the optimization problem LP$_1(\und{\alpha})$.
We can think of it as a minimal value $C$ for which the plane $C-\sum_{i=1}^m \alpha_i R_i$ intersects
the rate region $\mcl{R}$ defined in \eqref{cut_set}. It is not hard to conclude that the point
of intersection is one of the ``vertices'' of the region $\mcl{R}$, \emph{i.e.}, it is completely
defined by the collection of sets $\{\mcl{S}_1,\mcl{S}_2,\ldots,\mcl{S}_m\}$ such that
\begin{align}
R(\mcl{S}_i)&=H(X_{\mcl{S}_i}|X_{\mcl{S}_i^c}),~~~i\in \{1,2,\ldots,m\}. \nonumber
\end{align}
The following theorem will be very useful in Section \ref{sec:fls} when we explore the finite linear source model.
It represents a key building block for bounding the total number of breakpoints in $h(\beta)$.
\begin{theorem} \label{thm:bpts}
For every breakpoint of the function $h(\beta)$, the corresponding rate vector $\mbf{R}$ that minimizes \eqref{fcn:h}
is a vertex of the rate region $\mcl{R}$.
\end{theorem}
\begin{proof}
Due to the equivalence between the problems LP$_1(\und{\alpha})$ and LP$_3(\und{\alpha})$ it follows that for every $\und{\alpha}$, the rate tuple
$\mbf{R}$ which corresponds to the minimizer of the function $h(\beta)$, is a vertex of the rate region $\mcl{R}$.
For a given cost vector $\und{\alpha}$, we prove this theorem by modifying $\und{\alpha}$ such that each breakpoint in $h(\beta)$
can become the minimizer of the function $h$ that corresponds to the modified vector $\und{\alpha}$.

To that end, let us consider an example of $h(\beta)$ shown in Figure \ref{fig:bpts}.
Each linear segment of $h(\beta)$ is described by a pair of vector $(\mbf{b}^{(i)},\mbf{c}^{(i)})$, $i=1,2,3,4$, as in \eqref{eq:represent}.
Function $h(\beta)$ is minimized when $\beta=\beta_3$.
\begin{figure}[h]
\begin{center}
\includegraphics[scale=0.6]{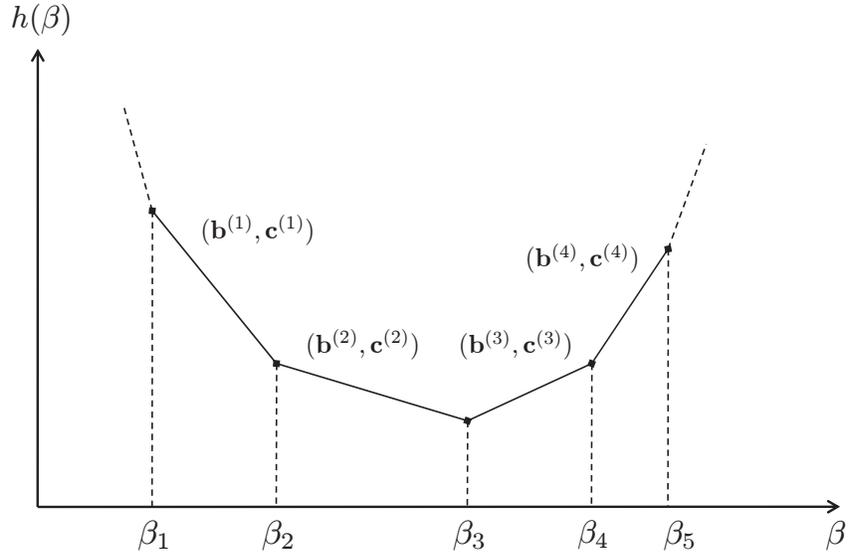}
\end{center}
\vspace{-0.2in}
\caption{Function $h(\beta)$ is a piecewise linear in $\beta$. For the purpose of proving Theorem \ref{thm:bpts},
         we consider $4$ linear segments, and show that each breakpoint can become the minimizer
         of a different optimization problem.}\label{fig:bpts}
\end{figure}

First, we show how to modify $\und{\alpha}$ so that the breakpoint $\beta_2$ becomes the minimizer of LP$_3(\und{\alpha})$.
From \eqref{eq:represent}, we have that the slopes of the segments $[\beta_1,\beta_2]$
and $[\beta_2,\beta_3]$ are such that $\sum_{i=1}^m \alpha_i b^{(1)}_i <0$, $\sum_{i=1}^m \alpha_i b^{(2)}_i <0$.
Since $h(\beta)$ is convex, it also holds that
\begin{align}
\sum_{i=1}^m \alpha_i b^{(1)}_i <\sum_{i=1}^m \alpha_i b^{(2)}_i. \label{eq:order}
\end{align}
Observe that for every $\beta\geq R_{CO}(\mbf{1})$, the rate tuple that corresponds to $h(\beta)$ is
such that $R(\mcl{M})=\beta$. Hence, for each linear segment it holds that $\sum_{j=1}^m b^{(i)}_j=1$,
$i=1,2,3,4$. Let
\begin{align}
\alpha'_i = \alpha_i + \Delta \alpha,~~i=1,2,\ldots,m, \label{eq:new_alpha}
\end{align}
where $\Delta \alpha>0$ is a constant. For the weight vector $\und{\alpha}'$ constructed in \eqref{eq:new_alpha},
the segments $[\beta_1,\beta_2]$ and $[\beta_2,\beta_3]$ have slopes
\begin{align}
\sum_{i=1}^m b^{(j)}_i (\alpha_i+\Delta \alpha) =  \sum_{i=1}^m b^{(j)}_i \alpha_i+\Delta \alpha,~~j=1,2. \nonumber
\end{align}
Therefore, we can pick $\Delta \alpha$ such that the linear segment $[\beta_1,\beta_2]$ has negative slope, while
the linear segment $[\beta_2,\beta_3]$ has positive slope. One possible choice is
\begin{align}
\Delta \alpha = - \sum_{i=1}^m \alpha_i b^{(2)}_i + \epsilon,
\end{align}
where $\epsilon$ is a small positive constant. Note that due to \eqref{eq:order} the linear segment
$[\beta_1,\beta_2]$ still has negative slope.
Similarly, we can move a minimizer of LP$_3(\und{\alpha})$ from $\beta_3$ to $\beta_4$, by modifying $\und{\alpha}$ as follows
\begin{align}
\alpha'_i = \alpha_i - \Delta \alpha,~~i=1,2,\ldots,m,
\end{align}
where $\Delta \alpha =  \sum_{i=1}^m \alpha_i b^{(3)}_i + \epsilon$. In this case, linear segments $[\beta_3,\beta_4]$
and $[\beta_4,\beta_5]$ have the slopes $\sum_{i=1}^m \alpha'_i b^{(3)}_i <0$ and $\sum_{i=1}^m \alpha'_i b^{(4)}_i >0$,
which makes $\beta=\beta_4$ the minimizer of LP$_3(\und{\alpha}')$.
Therefore, we showed how to modify the cost vector $\und{\alpha}$, so that the minimizer of $h(\beta)$ ``jumps''
to the consecutive breakpoints of $h(\beta)$. Repeating this procedure multiple times, one can modify $\und{\alpha}$ so that
any breakpoint becomes the minimizer of $h(\beta)$.
\end{proof}

\section{Data Exchange Problem with Linear Correlations} \label{sec:fls}

In this section we propose a polynomial time algorithm for achieving a rate tuple that belongs to the region $\mcl{R}(\und{\alpha})$
in the data exchange problem.
In Section \ref{sec:model} we defined a linear model where each user $i \in \mcl{M}$ observes
a collection of the linear equations in $\mathbb{F}_{q^n}$,
\begin{align}
X_i = \mathbf{A}_i \mbf{W}, \ i \in \mcl{M}, \label{mod:eq1}
\end{align}
where $\mathbf{A}_i \in \mathbb{F}_q^{\ell_i \times N}$ is a fixed matrix and $\mbf{W} \in \mathbb{F}_{q^n}^N$ is a vector of data packets.
Since all the algebraic operations are performed over the base field $\mathbb{F}_q$, the linear model \eqref{mod:eq1} is equivalent to the scenario
where each user observes $n$ memoryless instances of the finite linear process \eqref{mod:eq1} where $\mbf{W}$ is a uniform vector
over $\mathbb{F}_q^N$. Hereafter, we will use the entropy of the observations and the rank of the observation matrix interchangeably.

\begin{theorem}\label{thm:netcode}
For the linear source model, any rate tuple $\mbf{R}$ that belongs to the rate region $\mcl{R}_{de}$, defined in \eqref{cut_set_datexc},
can be achieved via linear network coding, \emph{i.e.}, in order to achieve omniscience
it is sufficient for each user $i \in \mcl{M}$ to transmit $R_i$ properly chosen linear equations of the data packets he observes.
\end{theorem}
Proof of Theorem~\ref{thm:netcode} is provided in Appendix~\ref{app:netcode}. This result suggests that in an optimal
communication scheme, each user transmits some integer number of symbols in $\mathbb{F}_q$.
Hence, a rate tuple that belongs to $\mcl{R}(\und{\alpha})$ in the data exchange problem has to be some fractional
number with the denominator $n$.
To that end, we introduce a fractional rate constraint to the optimization problem LP$_1(\und{\alpha})$
in order to obtain the optimal solution for the data exchange problem.
\begin{align}
\min_{\mbf{R}} \sum_{i=1}^m \alpha_i R_i,~~~\text{s.t.}~~R(\mcl{S})\geq H(X_{\mcl{S}}|X_{\mcl{S}^c}),~~\forall \mcl{S} \subset \mcl{M}, \label{int_lp1}
\end{align}
where $n \cdot R_i \in \mathbb{Z}$, $\forall i\in \mcl{M}$.
Optimization problem \eqref{int_lp1} is an integer linear program, henceforth denoted by ILP$_n(\und{\alpha})$.
We use $\mcl{R}_n(\und{\alpha})$ to denote the rate region of all minimizers of the above ILP, and $R_{CO,n}(\und{\alpha})$ to denote the minimal cost.

Notice that there is a certain gap between the ``information-theoretic'' optimal solution to the problem LP$_1(\und{\alpha})$,
and the ``data exchange'' optimal solution to the problem ILP$_n(\und{\alpha})$. The reason is that the former solution
assumes that the observation length tends to infinity, while in the data exchange setting we are dealing with the
finite block lengths.

In this section we show how to efficiently solve ILP$_n(\und{\alpha})$ by applying the optimization techniques
we derived so far. Then, we propose a polynomial time code construction based on the matrix completion method
over finite fields borrowed from the network coding literature \cite{H05}.

To gain more insight into the coding scheme, let us start with the problem of finding a rate tuple that
belongs to the region $\mcl{R}_n(\mbf{1})$.

\subsection{Achieving a rate tuple from $\mcl{R}_n(\mbf{1})$} \label{sec:1co}

Let us consider the optimization problem ILP$_n(\mbf{1})$.
Observe that by applying the modified Edmond's algorithm for any $\beta\geq R_{CO}(\mbf{1})$, we obtain a feasible rate tuple
that corresponds to the rate region $\mcl{R}_{de}$ defined in \eqref{cut_set_datexc}. Moreover, by setting $\beta$
to be a fractional number with the denominator $n$ in the problem LP$_1(\beta)$, we also get all the optimal rates to be
fractional numbers with the denominator $n$. Hence, an optimal rate tuple with respect to the optimization problem
ILP$_n(\mbf{1})$ can be obtained by applying the modified Edmond's algorithm for
$\beta = \frac{\lceil n \cdot R_{CO}(\mbf{1}) \rceil}{n} = R_{CO,n}(\mbf{1})$. The next natural question is
how far we are from the information-theoretic optimal solution, \emph{i.e.}, when $n \rightarrow \infty$.

\begin{claim}
The optimal sum rate w.r.t. ILP$_n(\mbf{1})$ is at most $\frac{1}{n}$ symbols in $\mathbb{F}_q$ away from $R_{CO}(\mbf{1})$.
\begin{align}
R_{CO,n}(\mbf{1})-R_{CO}(\mbf{1})\leq \frac{1}{n}.
\end{align}
\end{claim}

\begin{example} \label{example1}
{
Consider an example where $3$ users observe the packets of length $n=2$ over the field $\mathbb{F}_q$.
\begin{align}
\mbf{X}_1 = [\begin{array}{cc}
\mbf{a} & \mbf{b}
\end{array}], \nonumber \\
\mbf{X}_2 = [\begin{array}{cc}
\mbf{a} & \mbf{c}
\end{array}], \label{eq:example} \\
\mbf{X}_3 = [\begin{array}{cc}
\mbf{b} & \mbf{c}
\end{array}], \nonumber
\end{align}
where $\mbf{W}=\left[
                 \begin{array}{c}
                   \mbf{a} \\
                   \mbf{b} \\
                   \mbf{c} \\
                 \end{array}
               \right]$ is a data packet vector in $\mathbb{F}_{q^2}^3$ such that
$\mbf{a}=\left[
                 \begin{array}{cc}
                   a_1 & a_2
                 \end{array}
               \right]$,
$\mbf{b}=\left[
                 \begin{array}{cc}
                   b_1 & b_2
                 \end{array}
               \right]$,
$\mbf{c}=\left[
                 \begin{array}{cc}
                   c_1 & c_2
                 \end{array}
               \right]$.}

As pointed out above, we can this of this model as $n=2$ repetitions of the finite linear process.
Solving the problem ILP$_n(\mbf{1})$ for this example, we obtain $R_1=R_2=R_3=\frac{1}{2}$. Moreover, we
also obtain the same rate allocation for the LP$_1(\mbf{1})$, which suggests that in this case
there is no gap in optimality between the finite and infinite observation length.

In Theorem~\ref{thm:netcode} we showed that the network coding solution can achieve any rate tuple
that belongs to $\mcl{R}_{de}$, and hence, it also achieves any rate tuple from $\mcl{R}_n(\mbf{1})$.
It is not hard to see that one possible solution for this example is: user $1$ transmits $a_1+b_2$,
user $2$ transmits $c_1+a_2$, and user $3$ transmits $b_1+c_2$.
\end{example}

\subsection{Code Construction} \label{sec:code}

 The next question that arises from this analysis is how to design the actual transmissions of each user?
Starting from an optimal (integer) rate allocation, we construct the corresponding multicast network (see Figure~\ref{fig:multicast}).
Then, using polynomial time algorithms for the multicast code construction \cite{Jaggi05}, \cite{H05},
we can solve for the actual transmissions of each user.
We illustrate conversion of the data exchange problem to the multicast problem by considering the source model in Example~\ref{example1}.
Then, the extension to an arbitrary linear source model is straightforward.

In this construction, notice that there are $4$ different types of nodes.
Conversion of our problem into the multicast problem assumes the existence of the super-node, here denoted by $S$, that possess all
the packets. In the original problem, each user in the system plays the role of a transmitter and a receiver.
To distinguish between these two states, we denote
$s_1$, $s_2$ and $s_3$ to be the ``sending'' nodes, and $r_1$, $r_2$ and $r_3$ to be the ``receiving'' nodes
which corresponds to the users $1$, $2$ and $3$ in the original system, respectively.

Node $S$, therefore, feeds its information to the nodes $s_1$, $s_2$ and $s_3$.
Unlike the multicast problem, where any linear combination
of the packets can be transmitted from node $S$ to $s_1$, $s_2$ and $s_3$, here the transmitted packets
correspond to the observations of the users $1$, $2$ and $3$, respectively. The
second layer of the network is designed based on the optimal rates $R_1$, $R_2$ and $R_3$.
Since $n=2$, each user gets to transmit $1$ symbol in $\mathbb{F}_q$.
It is clear that all the receiving users are getting two different types of information:
\begin{enumerate}
\item The side information that each user already has. In the multicast network this information is transmitted
      directly from node $s_i$ to node $r_i$, $i=1,2,3$.
\item The information that each node $r_i$ receives from the other nodes $s_j$, $j\neq i$.
\end{enumerate}
To model the second type of information, let us consider the nodes $r_2$ and $r_3$.

Due to the broadcast nature of the channel, both $r_2$ and $r_3$ are receiving the same
symbol in $\mathbb{F}_{q}$ from node $s_1$. Thus, it is necessary to introduce a dummy
node $t_1$ to model this constraint. The capacities of the links $s_1-t_1$, $t_1-r_2$ and $t_1-r_3$
are all equal to $1$ symbol in $\mathbb{F}_{q}$. Note that this constraint ensures that the nodes
$r_2$ and $r_3$ are obtaining the same $1$ symbol from $s_1$. The remaining edges are designed in a similar way.

\begin{figure}
\begin{center}
\psset{unit=0.40mm}
\begin{pspicture}(0,0)(150,160)

\cnodeput[linestyle=solid,linecolor=black,fillstyle=solid](75,150){S}{$S$}

\cnodeput[linestyle=solid,linecolor=black,fillstyle=solid](5,90){A}{$s_1$}

\cnodeput[linestyle=solid,linecolor=black,fillstyle=solid](30,60){A2}{$t_1$}
\ncline[linestyle=solid]{-}{A}{A2}\Aput[0.9pt]{$1$}

\cnodeput[linestyle=solid,linecolor=black,fillstyle=solid](75,90){B}{$s_2$}

\cnodeput[linestyle=solid,linecolor=black,fillstyle=solid](90,60){B2}{$t_2$}
\ncline[linestyle=solid]{-}{B}{B2}\Aput[0.9pt]{$1$}

\cnodeput[linestyle=solid,linecolor=black,fillstyle=solid](145,90){C}{$s_3$}

\cnodeput[linestyle=solid,linecolor=black,fillstyle=solid](125,60){C2}{$t_3$}
\ncline[linestyle=solid]{-}{C}{C2}\Bput[0.9pt]{$1$}


\cnodeput[linestyle=solid,linecolor=black,fillstyle=solid](5,5){A1}{$r_1$}

\cnodeput[linestyle=solid,linecolor=black,fillstyle=solid](75,5){B1}{$r_2$}

\cnodeput[linestyle=solid,linecolor=black,fillstyle=solid](145,5){C1}{$r_3$}

\ncline{<-}{A}{S}\aput[0.8pt]{:U}{$a_1b_1a_2b_2$}
\ncline{<-}{B}{S}\aput[0.8pt]{:U}{$a_1c_1a_2c_2$}
\ncline{->}{S}{C}\aput[0.8pt]{:U}{$b_1c_1b_2c_2$}

\ncline{->}{A}{A1}\Bput[1pt]{$4$}
\ncline[linestyle=solid]{->}{A2}{B1}\bput[0.8pt](0.8){$1$}
\ncline[linestyle=solid]{->}{A2}{C1}\bput[0.8pt](0.8){$1$}

\ncline[linestyle=solid]{->}{B2}{A1}\bput[0.8pt](0.8){$1$}
\ncline{->}{B}{B1}\bput[1pt](0.2){$4$}
\ncline[linestyle=solid]{->}{B2}{C1}\aput[0.8pt](0.7){$1$}

\ncline[linestyle=solid]{->}{C2}{A1}\aput[0.8pt](0.8){$1$}
\ncline[linestyle=solid]{->}{C2}{B1}\aput[0.8pt](0.8){$1$}
\ncline{->}{C}{C1}\Aput[1pt]{$4$}

\end{pspicture}
\end{center}
\caption{Multicast network constructed from the source model and the
         optimal rate tuple $R_1=R_2=R_3=\frac{1}{2}$ that belongs to $\mcl{R}_2(\mbf{1})$. Each user receives side
         information from ``itself'' (through the links $s_i-r_i$, $i=1,2,3$) and the other users
         (through the links $t_i-r_j$, $i,j\in \{1,2,3\}$, $i\neq j$).} \label{fig:multicast}
\end{figure}
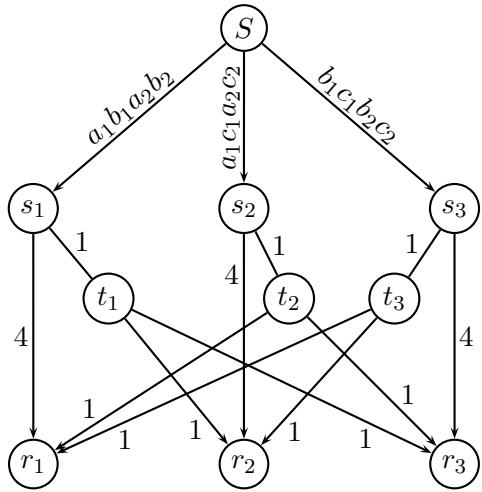

Now, when we have a well-defined network, it is only left to figure out transmissions on all the edges.
If we want to apply Jaggi's algorithm \cite{Jaggi05}, the first step
is to determine disjoint paths from the super-node $S$ to each receiver $r_1-r_3$ using the Ford-Fulkerson
algorithm \cite{B98}. While the solution to this
problem is easy in the case when each user observes only a subset of the packets (like in this example),
it is not trivial to find disjoint paths which connect linearly independent sources to the receivers $r_i$
(see Figure~\ref{fig:multicast}).  For that reason we apply Harvey's algorithm \cite{H05} which is based
on matrix representation of the transmissions in the network \cite{KM03}, \cite{ho2006random},
and simultaneous matrix completion problem over finite fields.

In \cite{KM03}, the authors derived
the transfer matrix $\mbf{M}(r_i)$ from the super-node $S$ to any receiver $r_i$, $i=1,2,\ldots,m$. It is a $N \times N$
matrix with the input vector $\mbf{W}$, and the output vector corresponding to the observations at the receiver $r_i$.
\begin{align}
\mbf{M}(r_i) = \mbf{A} (\mbf{I} - \mbf{\Gamma})^{-1} \mbf{B}(r_i),~~~i=1,2,\ldots,m,
\end{align}
where matrix $\mbf{A}$ is a source matrix, $\mbf{\Gamma}$ is adjacency matrix of the multicast network,
and $\mbf{B}(r_i)$ is an output matrix. For more details on how these matrices are constructed,
we refer the interested reader to the reference \cite{KM03}. Here, we just make a comment on the source matrix
$\mbf{A}$. In general, it is a $N \times \ell$ matrix, where $\ell$ is the total number of edges in the network.
Input to the matrix $\mbf{A}$ is the vector of independent packets $\mbf{W}$. For the source model in Figure~\ref{fig:multicast},
non-zero entries in the matrix $\mbf{A}$ correspond to the edges $S-s_1$, $S-s_2$ and $S-s_3$. Since, transmissions
on those edges are already assigned by the underlying source model, in general we have
\begin{align}
\mbf{A}=
\left[
  \begin{array}{ccccccc}
    \mbf{A}_1^T & \mbf{A}_2^T & \cdots & \mbf{A}_m^T & \mbf{0} & \cdots & \mbf{0}
  \end{array}
\right],
\end{align}
where $\mbf{A}_i$ corresponds to the observation matrix defined in \eqref{mod:eq1}.

Essentially, a multicast problem has a network coding solution if and only if each matrix $\mbf{M}(r_i)$
is non-singular. In \cite{H05}, the author showed that for the \emph{expanded transfer matrix} defined as
\begin{align}
\mbf{E}(r_i)=
\left[
  \begin{array}{cc}
    \mbf{A} & \mbf{0} \\
    \mbf{I}-\mbf{\Gamma} & \mbf{B}(r_i) \\
  \end{array}
\right],~~~i=1,2,\ldots,m,
\end{align}
it holds that $\det (\mbf{M}(r_i))=\pm \det (\mbf{E}(r_i))$.

It should be noted that some of the entries in matrices $\mbf{\Gamma}$ and $\mbf{B}(r_i)$, $i=1,2,\ldots,m$, are
unknowns. To obtain the actual transmissions on all the edges it is necessary to replace those unknown entries
with elements over $\mathbb{F}_q$ such that all matrices $\mbf{E}(r_i)$, $i=1,2,\ldots,m$, have full rank.
This is known as a simultaneous matrix completion problem and it is solved in \cite{H05} in polynomial time.

\begin{lemma}[Harvey, \cite{H05}] \label{lm:exits_sol}
Polynomial time solution for the simultaneous matrix completion problem exists if and only if $|\mathbb{F}_q|> m$.
The complexity of the proposed algorithm applied to the data exchange problem is
$\mcl{O}(m^4\cdot N^3 \cdot n^3 \cdot \log(m\cdot N \cdot n))$.
\end{lemma}


The complexity of the code construction can be further reduced when for the $(R_1,R_2,\ldots,R_m) \in \mcl{R}_n(\mbf{1})$
it holds that the greatest common divisor $gcd(nR_1,nR_2,\ldots,nR_m)>1$.
In this case, for every $\tilde{n}=\frac{n}{gcd(nR_1,nR_2,\ldots,nR_m)}$ generations of the finite linear process,
we still have that each user transmits some integer number of symbols in $\mathbb{F}_q$. Hence, it is enough
to construct a coding scheme for $\tilde{n}$ observations of the linear process, and then just to apply such scheme
$\frac{n}{\tilde{n}}$ times to solve the data exchange problem. From Lemma~\ref{lm:exits_sol} the complexity of such scheme
is $\mcl{O}(m^4\cdot N^3 \cdot \tilde{n}^3 \cdot \log(m\cdot N \cdot \tilde{n}))$.

\subsection{Asymptotic optimality of $R_{CO,n}(\mbf{1})$}

In this section we consider under which conditions there is no gap between the solution
of the problem ILP$_n(\mbf{1})$, when $n$ is finite, and the solution of LP$_1(\mbf{1})$
(asymptotic solution $n \rightarrow \infty$). To that end, let us consider the following Lemma.
\begin{lemma} \label{lm:rco}
Optimal $R_{CO}(\mbf{1})$ rate of the problem LP$_1(\mbf{1})$ can be expressed as
\begin{align}
R_{CO}(\mbf{1}) = H(X_{\mcl{M}})-\min_{\mcl{P}}\left\{\frac{\sum_{\mcl{S}\in \mcl{P}}H(X_{\mcl{S}})-H(X_{\mcl{M}})}{|\mcl{P}|-1}\right\},~~\text{$\mcl{P}$ is a partition of $\mcl{M}$ s.t. $|\mcl{P}|\geq 2$}. \label{lm1:rco}
\end{align}
\end{lemma}
Proof of Lemma \ref{lm:rco} is provided in Appendix \ref{app:lm:rco}. It is based on a geometry of the function $g(\mcl{M},\beta)$.
Minimization \eqref{lm1:rco} was also shown in \cite{CZ10} by considering an LP dual of the optimization problem LP$_1(\mbf{1})$.

From Lemma \ref{lm:rco}, $R_{CO}(\mbf{1})$ can be expressed as a rational number. Moreover, the denominator
of $R_{CO}(\mbf{1})$ can be some integer number between $1$ and $m-1$ depending on the cardinality of
the optimal partition according to \eqref{lm1:rco}. From Lemma \ref{lm:modedm} it immediately follows
that all $(R_1,R_2,\ldots,R_m) \in \mcl{R}_n(\mbf{1})$ are also rational numbers with the denominator $n$.

To that end, if $n$ is divisible by $|\mcl{P}(R_{CO}(\mbf{1}))-1|$ for $\mcl{P}(R_{CO}(\mbf{1}))|\geq 2$, then
\begin{align}
R_{CO,n}(\mbf{1})=R_{CO}(\mbf{1}).
\end{align}

\subsection{Achieving a rate tuple from $\mcl{R}_n(\und{\alpha})$} \label{sec:rco_alpha}

In Section~\ref{sec:code} we argued that once we obtain the optimal fractional rates (which denote how many symbols in $\mathbb{F}_q$
each user transmits), the construction of the corresponding
multicast network is straightforward, and hence, the coding scheme can be obtained in polynomial time by
using the algorithm proposed in \cite{H05}.
Here, we describe an algorithm that finds an optimal solution to the optimization problem ILP$_n(\und{\alpha})$.


In Section~\ref{sec:alpha} we proposed the gradient descent algorithm to achieve an approximate solution to the problem LP$_1(\und{\alpha})$.
Setting the precision parameter $\varepsilon=\frac{1}{n}$ it is guaranteed that the distance between the sum rate
which corresponds to the rate tuple from  $\mcl{R}(\und{\alpha})$ and the sum rate obtained through the gradient descent algorithm,
is at most $\frac{1}{n}$, \emph{i.e.}, $|\beta_{gd}-\beta^{\star}|\leq \frac{1}{n}$.
Therefore, we have
\begin{align}
|n \beta_{gd}-n \beta^{\star}|\leq 1. \label{eq:beta_bound1}
\end{align}
From \eqref{eq:beta_bound1} we conclude that
\begin{align}
\left| \frac{\lfloor n \beta_{gd}\rfloor}{n}-\beta^{\star} \right|\leq \frac{1}{n}~~\text{or}~~
\left| \frac{\lceil n \beta_{gd}\rceil}{n}-\beta^{\star} \right|\leq \frac{1}{n}. \label{eq:beta_bound2}
\end{align}
From \eqref{eq:beta_bound2} it follows that we can achieve a rate tuple from $\mcl{R}_n(\und{\alpha})$
which sum rate is at most $\frac{1}{n}$ away from $\beta^{\star}$ by choosing $\beta=\frac{\lfloor n \beta_{gd}\rfloor}{n}$ or
$\beta=\frac{\lceil n \beta_{gd}\rceil}{n}$. Let us denote by $\beta_{(n)}$ the optimal sum rate w.r.t. ILP$_n(\und{\alpha})$.
To decide which one of the proposed $\beta$'s is equal to $\beta_{(n)}$, we just need to compare the values of the function $h$ at these points.
\begin{align}
\beta_{(n)} = \arg\min \left\{h\left( \frac{\lfloor n\beta_{gd} \rfloor}{n} \right), h\left( \frac{\lceil n\beta_{gd} \rceil}{n} \right)  \right\}. \label{eq:beta_bound3}
\end{align}
Then, it follows that
\begin{align}
|R_{CO,n}(\und{\alpha})-R_{CO}(\und{\alpha})|\leq \frac{\max_i \alpha_i}{n}.
\end{align}
Complexity of the proposed algorithm is $\mcl{O}(m^2\cdot SFM(m) + \log(n\cdot N)\cdot m \cdot SFM(m))$.
After obtaining an optimal communication rates w.r.t. ILP$_n(\und{\alpha})$, it is only left to
apply the code construction algorithm proposed in Subsection~\ref{sec:1co} (see Lemma~\ref{lm:exits_sol}).

\subsection{Asymptotic optimality of $R_{CO,n}(\und{\alpha})$}

In this section we explore under which conditions the optimal solutions of the problems
ILP$_n(\und{\alpha})$ and LP$_1(\und{\alpha})$ are the same.

In order to obtain the asymptotically optimal rates w.r.t. LP$_1(\und{\alpha})$,
it is necessary to bound from bellow the length of each linear segment in $h(\beta)$. Then, by
choosing the appropriate step size in the gradient descent algorithm, we can achieve the goal.

\begin{theorem}\label{thm:bound}
An optimal asymptotic solution to the problem LP$_1(\und{\alpha})$ in the finite linear source model
can be obtained in polynomial time by using a gradient descent method with the precision parameter $\varepsilon=m^{-m/2}$.
Complexity of the proposed algorithm is $\mcl{O}((m\cdot \log m+ \log N) \cdot m \cdot SFM(m))$.
\end{theorem}
\begin{proof}
In Theorem \ref{thm:bpts} we showed that each breakpoint in $h(\beta)$ corresponds to a vertex of the rate region $\mcl{R}$
defined in \eqref{cut_set}. In other words, for some breakpoint $\beta_j$, the optimal rate tuple is uniquely defined
by the following system of equations
\begin{align}
R(\mcl{S}_i) = H(X_{\mcl{S}_i}|X_{\mcl{S}^c_i}),~~i=1,2,\ldots,m, \label{system_eq}
\end{align}
where $\mcl{S}_i \subset \mcl{M}$. Moreover, it holds that $R(\mcl{M})=\beta_r$.
System of linear equations \eqref{system_eq} can be expressed in a matrix form as follows.
\begin{align}
\mbf{\Lambda}\cdot \mbf{R} =
\left[
  \begin{array}{cccc}
    H(X_{\mcl{S}_1}|X_{\mcl{S}^c_1}) & H(X_{\mcl{S}_2}|X_{\mcl{S}^c_2}) & \ldots & H(X_{\mcl{S}_m}|X_{\mcl{S}^c_m}) \\
  \end{array}
\right]^T, \label{eq:matrix}
\end{align}
where
\begin{align}
\mbf{\Lambda}_{i,j}=
\begin{cases}
1  &  \text{if}~j\in \mcl{S}_i, \\
0  &  \text{otherwise}. \\
\end{cases}
\end{align}
In order to obtain the optimal rate tuple which corresponds to the breakpoint $\beta_r$, we can simply invert the matrix $\mbf{\Lambda}$.
Notice that the right hand side of \eqref{eq:matrix} consists of the conditional entropy (rank) expressions, which are, in the
case of the linear source model, integers. Therefore, all optimal rates $\mbf{R}$ which correspond to the breakpoints of $h(\beta)$
are fractional numbers with the denominator equal to the $\det(\mbf{\Lambda})$. This comes from the fact that
\begin{align}
\mbf{\Lambda}^{-1}=\frac{1}{\det(\mbf{\Lambda})}\cdot \text{adj}(\mbf{\Lambda}),
\end{align}
where adj$(\mbf{\Lambda})$ is the adjugate  of $\mbf{\Lambda}$. From \cite{hadamard} it follows that
\begin{align}
|\det(\mbf{\Lambda})|\leq m^{m/2}. \label{det:bound}
\end{align}
Therefore, all the breakpoints of $h(\beta)$ are at the distance of at least $m^{-m/2}$ from each other. Hence,
by setting the precision parameter in a gradient descent algorithm to $\varepsilon=m^{-m/2}$, we can make sure
that the minimizer of $h(\beta)$ is the end point of the linear segment to which approximate solution belongs to.
\end{proof}

In the further text, we explain how to find the minimum of $h(\beta)$ by applying a simple binary search algorithm
on top of the gradient descent algorithm proposed in Theorem \ref{thm:bound}. Let us consider the scenario in Figure \ref{fig:balg}.
Applying the gradient descent algorithm, with the precision parameter $\varepsilon=m^{-m/2}$ we can reach a point $\beta_{gd}$ that is $\varepsilon$
close to the minimizer $\beta^{\star}$ of $h(\beta)$, \emph{i.e.}, $|\beta_{gd}-\beta^{\star}|\leq m^{-m/2}$. Applying the modified Edmond's algorithm for $\beta=\beta_{gd}$,
we obtain parameters $(\mbf{b}^{(gd)},\mbf{c}^{(gd)})$ (see \eqref{eq:represent}) which correspond to the linear segment to which $\beta_{gd}$ belongs to.
In order to obtain $\beta^{\star}$ we simply need to jump to the consecutive linear segment. To that end, let $\beta_1=\beta_{gd}-m^{-m/2}$ belongs to the
linear segment $(\mbf{b}^{(1)},\mbf{c}^{(1)})$. Then, $\beta^{\star}$ can be obtained by intersecting these two linear segments.
\begin{align}
\beta^{\star} = \frac{\sum_{i=1}^m c^{(1)}_i \alpha_i-\sum_{i=1}^m c^{(gd)}_i \alpha_i}{\sum_{i=1}^m b^{(gd)}_i \alpha_i - \sum_{i=1}^m b^{(1)}_i \alpha_i}.
\end{align}
\begin{figure}[h]
\begin{center}
\includegraphics[scale=0.6]{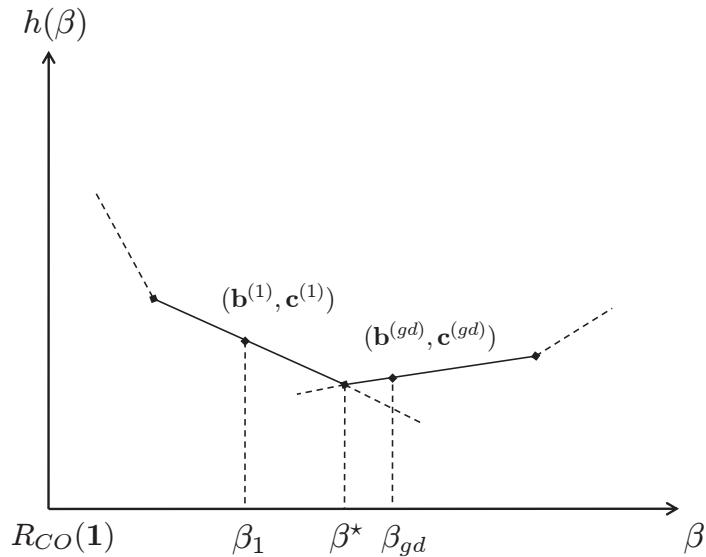}
\end{center}
\vspace{-0.2in}
\caption{Line intersection procedure applied on top of the gradient descent algorithm to obtain the minimum $h(\beta^{\star})$.}\label{fig:balg}
\end{figure}

Therefore, if the data packet length $n$ is divisible by the denominator of $\beta^{\star}$, then $R_{CO,n}(\und{\alpha})=R_{CO}(\und{\alpha})$.

\subsection{Indivisible Packets}

Let us now consider the scenario when the data packets cannot be split.
To obtain an optimal communication rates, we can directly apply the results form the Sections \ref{sec:1co} and \ref{sec:rco_alpha}.
We can think of this problem as having one packet over very large base field $\mathbb{F}_{q^n}$.

Hence, for the case when $\und{\alpha}=\mbf{1}$, it holds that
\begin{align}
R_{CO,1}(\mbf{1})=\lceil R_{CO}(\mbf{1}) \rceil~~\text{symbols in $\mathbb{F}_{q^n}$.} \nonumber
\end{align}
Similarly, we can obtain the sum rate which corresponds to the optimal $R_{CO,1}(\und{\alpha})$ as follows.
\begin{align}
\beta_{(1)} = \arg\min \left\{h\left( \lfloor \beta_{gd} \rfloor \right), h\left( \lceil \beta_{gd} \rceil \right)  \right\}~~\text{symbols in $\mathbb{F}_{q^n}$.} \nonumber
\end{align}
However, in the actual coding scheme, all the algebraic operations are performed over the original base field $\mathbb{F}_q$.

\section{Conclusion} \label{sec:conclusion}

In this work we addressed the problem of the data exchange, where each user in the system possess some partial knowledge
(side information) about the file that is of common interest. The goal is for each user to gain access to the entire file while
minimizing the (possibly weighted) amount of bits that these users need to exchange over a noiseless public channel.
For the general case when the side information is in form of the i.i.d. realizations of some discrete memoryless process, 
we provide a polynomial time algorithm that finds an optimal rate allocation w.r.t. communication cost. 
Our solution is based on some combinatorial optimization techniques such as optimizations over submodular polyhedrons, 
Dilworth truncation of intersecting submodular functions, Edmond's greedy algorithm, etc.  
For the case when the side information is in form of the linearly coded packets, besides an optimal rate allocation in polynomial time, 
we provide efficient methods for constructing linear network codes that can achieve omniscience among the users at the optimal rates
with finite block lengths and zero-error.

\appendices
\section{Proof of Lemma \ref{lm:app_opt1}}
\label{app:lm:app_opt1}
Base polyhedron $B(f,\leq)$ is defined by the following system of inequalities
\begin{align}
Z(\mcl{S})&\leq f(\mcl{S}),~~~\mcl{S} \subset \mcl{M}, \\
Z(\mcl{M})&=f(\mcl{M}).
\end{align}
This is equivalent to the following
\begin{align}
Z(\mcl{S}^c)&\geq f^{\star}(\mcl{S}^c)(=f(\mcl{M})-f(\mcl{S})), \\
Z(\mcl{M})&=f^{\star}(\mcl{M})(=f(\mcl{M})),
\end{align}
where the last equality holds because $f(\emptyset)=0$.
For the second part, we have
\begin{align}
(f^{\star})^{\star}(\mcl{S})&=f^{\star}(\mcl{M})-f^{\star}(\mcl{S}^c) \nonumber \\
                            &=f(\mcl{M})-(f(\mcl{M})-f(\mcl{S}))=f(\mcl{S}). \nonumber
\end{align}
\section{Proof of Lemma \ref{lm:f}}
\label{app:lm:f}
Using the properties of conditional entropy, we can write $f(\mcl{S},\beta)=\beta-H(X_{\mcl{M}})+H(X_{\mcl{S}})$.
When $\mcl{S}\cap \mcl{T}\neq \emptyset$, then the following inequality holds due to submodularity of entropy
\begin{align}
f(\mcl{S},\beta)+f(\mcl{T},\beta) &= H(X_{\mcl{S}})+H(X_{\mcl{T}})-2(H(X_{\mcl{M}})-\beta) \nonumber  \\
&~~\geq H(X_{\mcl{S}\cup \mcl{T}})+H(X_{\mcl{S}\cap \mcl{T}})-2(H(X_{\mcl{M}})-\beta) = f(\mcl{S}\cup \mcl{T},\beta) + f(\mcl{S}\cap \mcl{T},\beta). \label{submod}
\end{align}
Inequality \eqref{submod} holds whenever $\mcl{S}\cap \mcl{T}\neq \emptyset$. To show that the function $f$ is submodular
when $\beta\geq H(X_{\mcl{M}})$ it is only left to consider the case $\mcl{S}\cap \mcl{T} = \emptyset$. Since $f(\emptyset,\beta)=0$, we have
\begin{align}
f(\mcl{S},\beta)+f(\mcl{T},\beta)&=H(X_{\mcl{S}})+H(X_{\mcl{T}})-2(H(X_{\mcl{M}})-\beta) \nonumber  \\
&~~\geq H(X_{\mcl{S}},X_{\mcl{T}})-(H(X_{\mcl{M}})-\beta)= f(\mcl{S}\cup \mcl{T},\beta). \label{submod1}
\end{align}
Inequality in \eqref{submod1} follows from the fact that
\begin{align}
H(X_{\mcl{S}})+H(X_{\mcl{T}})-H(X_{\mcl{S},\mcl{T}})=I(X_{\mcl{S}};X_{\mcl{T}})\geq 0 \geq \beta - H(X_{\mcl{M}}).
\end{align}
This completes the proof.

\section{Proof of Lemma \ref{lm:nondec}} \label{app:lm:nondec}

Let us define function $g(\mcl{M},\beta,i)$, $i=1,2,\ldots,m$ as follows
\begin{align}
g(\mcl{M},\beta,i) = \min_{\mcl{P}} \left\{ \sum_{\mcl{S}\in \mcl{P}} \beta-H(X_{\mcl{S}^c}|X_{\mcl{S}}),~~\text{s.t.}~~|\mcl{P}|=i :
\text{$\mcl{P}$ is a partition of $\mcl{M}$}\right\}. \label{lm:eq_part}
\end{align}
Function $g(\mcl{M},\beta,i)$ is linear in $\beta$ for any fixed $i=1,2,\ldots,m$. Then, the Dilworth truncation $g(\mcl{M},\beta)$ can be written as
\begin{align}
g(\mcl{M},\beta)=\min_{i=1,2,\ldots,m} g(\mcl{M},\beta,i). \label{lm:eq_part1}
\end{align}
Note that the minimization \eqref{lm:eq_part} does not depend on $\beta$ since it can be written as
\begin{align}
g(\mcl{M},\beta,i) = i(\beta-H(X_{\mcl{M}})) + \min_{\mcl{P}} \left\{ \sum_{\mcl{S}\in \mcl{P}} H(X_{\mcl{S}}),~~\text{s.t.}~~|\mcl{P}|=i :
\text{$\mcl{P}$ is a partition of $\mcl{M}$}\right\}. \label{lm:eq_part_mod}
\end{align}
\begin{figure}[h]
\begin{center}
\includegraphics[scale=0.6]{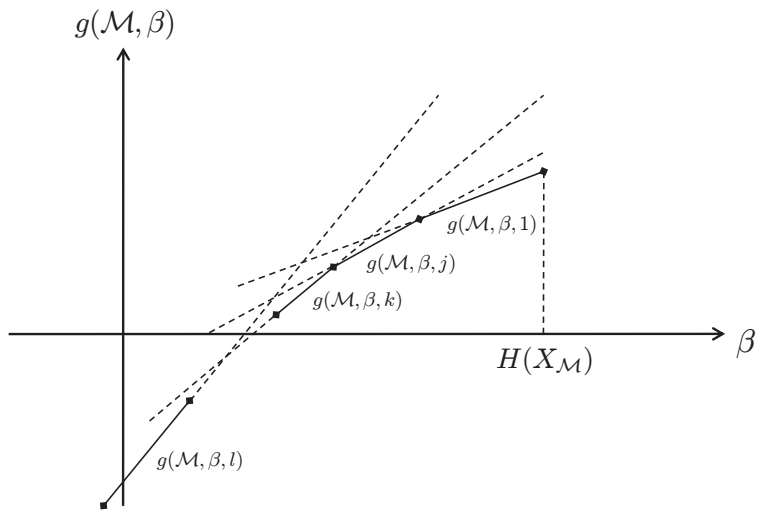}
\end{center}
\vspace{-0.2in}
\caption{Function $g(\mcl{M},\beta)$ is piecewise linear in $\beta$. It can be obtained by minimization over $m$ linear
         functions $g(\mcl{M},\beta,i)$, $i=1,2,\ldots,m$. $g(\mcl{M},\beta)$ has non-increasing slope, \emph{i.e.},
         $1\leq j\leq k \leq \cdots \leq l\leq m$.}\label{fig:g}
\end{figure}
Therefore, $g(\mcl{M},\beta)$ can be solved for any given $\beta$ by minimizing over all $m$ lines $g(\mcl{M},\beta,i)$, $i=1,2,\ldots,m$.
Hence, $g(\mcl{M},\beta)$ has at most $m$ linear segments. Moreover, due to minimization \eqref{lm:eq_part1}, $g(\mcl{M},\beta)$
has non-increasing slope (see Figure \ref{fig:g}).

To verify that the last linear segment in $g(\mcl{M},\beta)$ is of slope $1$, it is sufficient to find a point $\beta$
for which the function $g(\mcl{M},\beta)$ has slope $1$. To that end, let us consider $\beta=H(X_{\mcl{M}})$. From Lemma \ref{lm:f} it
follows that $f(\mcl{S},\beta=H(X_{\mcl{M}}))$ is submodular function, and hence,
$g(\mcl{M},\beta=H(X_{\mcl{M}})) = f(\mcl{M},\beta=H(X_{\mcl{M}})) = \beta$, where the last equality follows from \eqref{fcn:f}.
Therefore, the slope of $g(\mcl{M},\beta)$ at $\beta=H(X_{\mcl{M}})$ is $1$, which completes the proof.

\section{Optimal Partitioning w.r.t. Dilworth Truncation}\label{app:dilw}

In \cite{NKI10} it was shown how to obtain an optimal partition $\mcl{P}(\beta)$ of the set $\mcl{M}$ w.r.t. \eqref{dilw1} from the modified Edmond's algorithm.
Here we provide intuition behind these results.
From Remark~\ref{rmk:modedm_base} it follows that $g(\mcl{M},\beta)$ is the optimal value of the optimization problem LP$_2(\beta)$.
As we pointed out in Section~\ref{sec:comb}, in each iteration $i$ of the modified  Edmond's algorithm,
we obtain a set $\mcl{S}_i$
for which the inequality constraint in $P(f,\beta)$ holds with equality. In the next claim we state a result that is crucial
for obtaining an optimal partition of $\mcl{M}$ with respect to Dilworth truncation of $f(\mcl{M},\beta)$.
\begin{claim}\label{cl2}
For an optimal solution $\mbf{Z}$ of the problem LP$_2(\beta)$, if $Z(\mcl{S}_1)=f(\mcl{S}_1)$, and
$Z(\mcl{S}_1)=f(\mcl{S}_2)$ then $Z(\mcl{S}_1 \cup \mcl{S}_2)=f(\mcl{S}_1 \cup \mcl{S}_2)$.
\end{claim}
\begin{proof}
For an optimal rate vector $\mbf{Z}$ of the problem LP$_2(\beta)$ we have
\begin{align}
Z(\mcl{S}_i) &= \beta-H(X_{\mcl{S}_i^c}|X_{\mcl{S}_i}) = \beta-H(X_{\mcl{M}})+H(X_{\mcl{S}_i}), \label{app:eq1} \\
Z(\mcl{S}_j) &= \beta-H(X_{\mcl{S}_j^c}|X_{\mcl{S}_j}) = \beta-H(X_{\mcl{M}})+H(X_{\mcl{S}_j}). \label{app:eq2}
\end{align}
Since LP$_2(\beta)$ represents optimization over the polyhedron $P(f,\mcl{\beta})$ it holds that
\begin{align}
Z(\mcl{S}_i \cup \mcl{S}_j) & \leq \beta - H(X_{\mcl{M}}) + H(X_{\mcl{S}_i},X_{\mcl{S}_j}), \label{app:eq3} \\
Z(\mcl{S}_i \cap \mcl{S}_j) & \leq \beta - H(X_{\mcl{M}}) + H(X_{\mcl{S}_i \cap \mcl{S}_j}). \label{app:eq4}
\end{align}
From \eqref{app:eq1} and \eqref{app:eq2} it follows that
\begin{align}
Z(\mcl{S}_i \cup \mcl{S}_j) &= Z(\mcl{S}_i)+Z(\mcl{S}_j)-Z(\mcl{S}_i \cap \mcl{S}_j) \nonumber \\
                            &= \beta-H(X_{\mcl{M}})+H(X_{\mcl{S}_i})+\beta-H(X_{\mcl{M}})+H(X_{\mcl{S}_j}) - Z(\mcl{S}_i \cap \mcl{S}_j) \nonumber \\
                            &\geq \beta-H(X_{\mcl{M}})+H(X_{\mcl{S}_i})+H(X_{\mcl{S}_j})-H(X_{\mcl{S}_i \cap \mcl{S}_j}), \label{app:eq5}
\end{align}
where the last step in \eqref{app:eq5} follows from \eqref{app:eq4}. Due to submodluarity of entropy it directly follows from \eqref{app:eq5} that
\begin{align}
Z(\mcl{S}_i \cup \mcl{S}_j) \geq \beta - H(X_{\mcl{M}}) + H(X_{\mcl{S}_i},X_{\mcl{S}_j}). \label{app:eq6}
\end{align}
Comparing \eqref{app:eq3} and \eqref{app:eq6} it must hold that
\begin{align}
Z(\mcl{S}_i \cup \mcl{S}_j) = \beta - H(X_{\mcl{M}}) + H(X_{\mcl{S}_i},X_{\mcl{S}_j}).
\end{align}
\end{proof}

Results of Claim \ref{cl2} represent a key building block for obtaining an optimal partition $\mcl{P}(\beta)$
for some fixed $\beta$ (see Algorithm \ref{alg:partition}).
From Remark~\ref{rmk:modedm_base} it follows that for the maximizer rate vector $\mbf{Z}$ of the problem LP$_2(\beta)$ it holds that
\begin{align}
Z(\mcl{S})=f(\mcl{S},\beta),~~\forall \mcl{S}\in \mcl{P}(\beta). \label{eq:tight}
\end{align}
From Claim \ref{cl2} and \eqref{eq:tight}
it follows that for the sets $S_i$ and $S_j$, which are the minimizer sets in iterations $i$ and $j$ of the modified Edmond's algorithm,
if $\mcl{S}_i \cap \mcl{S}_j \neq \emptyset$, then $\mcl{S}_i \cup \mcl{S}_j$ is a subset of the one of the
partition sets in $\mcl{P}(\beta)$. Therefore, in each iteration of the modified Edmond's algorithm, whenever the minimizer set
intersects some of the previously obtained sets, they must all belong to the same partition set (see steps $4$ and $5$ in Algorithm \ref{alg:partition}).
\begin{algorithm}
\caption{Optimal Partition \cite{NKI10}}
\label{alg:partition}
\begin{algorithmic}[1]
\STATE Let $j(1),j(2),\ldots,j(m)$ be any ordering of $\{1,2,\ldots,m\}$, and $\mcl{A}_i=\{j(1),j(2),\ldots,j(i)\}$.
\STATE Initialize $\mcl{P}^{0}=0$.
\FOR {$i=1$ to $m$}
\STATE Let $\mcl{S}_i$ be the minimizer of
       \begin{align}
       Z_{j(i)}&=\min \{f(\mcl{S},\beta)-Z(\mcl{S}) : j(i)\in \mcl{S},~~\mcl{S}\subseteq \mcl{A}_i \}. \nonumber
       \end{align}
\STATE $\mcl{T}_i = \mcl{S}_i \cup [\cup\{\mcl{V} : \mcl{V}\in \mcl{P}^{i-1},~\mcl{S}_i \cap \mcl{V} \neq \emptyset\}]$
\STATE $\mcl{P}^{i} = \{\mcl{T}_i\} \cup \{\mcl{V}: \mcl{V}\in \mcl{P}^{i-1},~\mcl{S}_i \cap \mcl{V} = \emptyset\}$
\ENDFOR
\STATE $\mcl{P}(\beta) = \mcl{P}^{m}$.
\end{algorithmic}
\end{algorithm}
Algorithm \ref{alg:partition} compared to the modified Edmond's Algorithm, has two additional steps in each iteration (step $5$ and step $6$).
Thus, the order of complexity of both algorithms is the same  and it is $\mcl{O}(m \cdot SFM(m))$.
The complete explanation of the Algorithm \ref{alg:partition} can be found in \cite{NKI10}. 
\section{Proof of Lemma~\ref{lm:h_prop}} \label{app:lmh}

Function $h(\beta)$ is given by
\begin{align}
h(\beta) = \min_{\mbf{R}} \sum_{i=1}^m \alpha_i R_i~~~\text{s.t.}~~R(\mcl{M})=\beta,~~R(\mcl{S})\geq H(X_{\mcl{S}}|X_{\mcl{S}^c}),
~~\forall \mcl{S}\subset \mcl{M}. \label{eq:mlp3}
\end{align}

\subsection*{Continuity of $h(\beta)$}

Let the rate tuple $(R^{(1)}_1,R^{(1)}_2,\ldots,R^{(1)}_m)$ corresponds to the minimizer $\beta_1$ of the function
$h(\beta)$, \emph{i.e.}, $\sum_{i=1}^m R^{(1)}_i=\beta_1$.
Then, for a point $\beta_2=\beta_1+\Delta \beta$ let us construct the rate tuple
\begin{align}
R^{(2)}_i=
\begin{cases}
R^{(1)}_i+\Delta \beta   &   \text{if $\beta=1$}, \\
R^{(1)}_i                &   \text{if $\beta \neq 1$}. \\
\end{cases}
\end{align}
Then $(R^{(2)}_1,R^{(2)}_2,\ldots,R^{(2)}_m)$ is a feasible rate tuple for the optimization problem \eqref{eq:mlp3}
when $\sum_{i=1}^m R^{(2)}_i=\beta_2$.
Moreover, it holds that $h(\beta_2)-h(\beta_1)\leq \alpha_1 \Delta \beta$. Hence,
\begin{align}
|\beta_2-\beta_1|\leq \Delta \beta \Rightarrow |h(\beta_2)-h(\beta_1)|\leq \alpha_1 \Delta \beta,
\end{align}
Since $\alpha_1<\infty$ by the model assumption, it immediately follows that the function $h(\beta)$ is continuous.

\subsection*{Convexity of $h(\beta)$}

Consider two points $\beta_1$ and $\beta_2$ such that $\beta_i \geq R_{CO}(\mbf{1})$, $i=1,2$. We want
to show that for any $\lambda\in [0,1]$ it holds that
$h(\lambda\beta_1+(1-\lambda)\beta_2)\leq \lambda h(\beta_1)+(1-\lambda) h(\beta_2)$. To that end, let
$\mbf{R}^{(1)}$ and $\mbf{R}^{(2)}$ be the optimal rate tuples w.r.t. $h(\beta_1)$ and
$h(\beta_2)$, respectively. Now, we show that $\mbf{R}=\lambda \mbf{R}^{(1)}+(1-\lambda)\mbf{R}^{(2)}$
is feasible rate tuple for the problem \eqref{eq:mlp3} when $\beta = \lambda\beta_1+(1-\lambda) \beta_2$.

Since $R^{(1)}(\mcl{M})=\beta_1$ and $R^{(2)}(\mcl{M})=\beta_2$, it follows that
\begin{align}
R(\mcl{M})=\lambda R^{(1)}(\mcl{M})+(1-\lambda) R^{(2)}(\mcl{M})=\lambda\beta_1+(1-\lambda) \beta_2. \label{lmc:1}
\end{align}
Since $R^{(1)}(\mcl{S})\geq H(X_{\mcl{S}}|X_{\mcl{S}^c})$, $R^{(2)}(\mcl{S})\geq H(X_{\mcl{S}}|X_{\mcl{S}^c})$,
$\forall \mcl{S} \subset \mcl{M}$, we have
\begin{align}
R(\mcl{S})=\lambda R^{(1)}(\mcl{S}) + (1-\lambda) R^{(2)}(\mcl{S}) \geq H(X_{\mcl{S}}|X_{\mcl{S}^c}). \label{lmc:2}
\end{align}
From \eqref{lmc:1} and \eqref{lmc:2} it follows that $\mbf{R}$ is a feasible rate tuple
w.r.t. optimization problem \eqref{eq:mlp3}. Therefore,
$\sum_{i=1}^m \alpha_i R_i \geq h(\lambda\beta_1+(1-\lambda)\beta_2)$. Hence,
\begin{align}
h(\lambda\beta_1+(1-\lambda)\beta_2)\leq \lambda h(\beta_1)+(1-\lambda) h(\beta_2),
\end{align}
which completes the proof.

\section{Proof of Lemma \ref{lm:rco}}
\label{app:lm:rco}
For $\beta=R_{CO}(\mbf{1})$ it holds that $|\mcl{P}(\beta)|=1$. Since $\beta=R_{CO}(\mbf{1})$ is also a breakpoint in $g(\mcl{M},\beta)$
(see Lemma \ref{lm:nondec}),
we have that $|\mcl{P}(\beta)|\geq 2$. In other words, optimal partition of the set $\mcl{M}$ is not unique. From \eqref{dilw1}
and \eqref{eq:rco} we can write expression for $R_{CO}(\mbf{1})$ as follows
\begin{align}
R_{CO}(\mbf{1})= |\mcl{P}(R_{CO}(\mbf{1}))| R_{CO}(\mbf{1})-\sum_{\mcl{S}\in \mcl{P}(R_{CO}(\mbf{1}))} H(X_{\mcl{S}^c}|X_{\mcl{S}}). \label{lm:eq}
\end{align}
Rearranging terms in \eqref{lm:eq} we get
\begin{align}
(|\mcl{P}(R_{CO}(\mbf{1}))|-1)R_{CO}(\mbf{1})=\sum_{\mcl{S}\in \mcl{P}(R_{CO}(\mbf{1}))} H(X_{\mcl{S}})
                                               -|\mcl{P}(R_{CO}(\mbf{1}))| H(X_{\mcl{M}}).
\end{align}
Dividing both sides of equality by $(|\mcl{P}(R_{CO}(\mbf{1}))|-1)$ we obtain
\begin{align}
R_{CO}(\mbf{1}) = H(X_{\mcl{M}})
                 -\frac{\sum_{\mcl{S}\in\mcl{P}(R_{CO}(\mbf{1}))}H(X_{\mcl{S}})-H(X_{\mcl{M}})}{|\mcl{P}(R_{CO}(\mbf{1}))|-1}.
\end{align}
This completes the proof of \eqref{lm1:rco} since $|\mcl{P}(R_{CO}(\mbf{1}))|\geq 2$.
\section{Proof of Theorem \ref{thm:netcode}} \label{app:netcode}

We prove this theorem by showing that for any rate tuple $\mbf{R}$ that belongs to the rate region $\mcl{R}_{de}$, defined in \eqref{cut_set_datexc},
there exists a network coding solution to the data exchange problem.

In the data exchange problem, each of the $m$ users get to observe some collection of linear combinations of the data packets
$w_1,w_2,\ldots,w_N$.
\begin{align}
X_i = \mbf{A}_i \cdot \mbf{W},~~\forall i \in \mcl{M},
\end{align}
where $\mbf{A}_i \in \mathbb{F}_p^{\ell_i \times N}$, and $\mbf{W}=\left[
                 \begin{array}{cccc}
                   w_1 & w_2 & \ldots & w_N \\
                 \end{array}
               \right]^T \in \mathbb{F}_{p^k}^N$.
\begin{figure}[h]
\begin{center}
\includegraphics[scale=0.52]{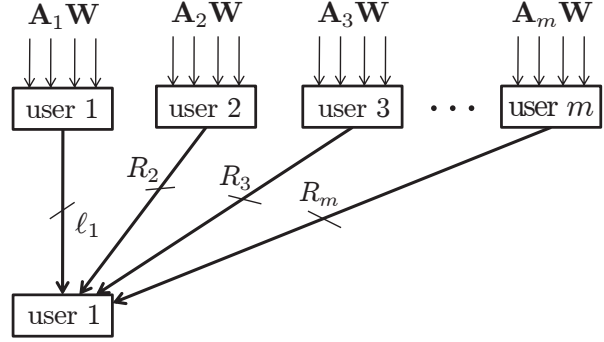}
\end{center}
\vspace{-0.2in}
\caption{Data exchange problem can be interpreted as a multicast problem. Considering user $1$ as a receiver, it obtains
the side information from himself through the link of capacity $\ell_1$, and it receives transmissions from the other
users through the links of capacities $R_i$, $i=2,3,\ldots,m$.} \label{fig:netcode}
\end{figure}

Since each user is interested in recovering all the data packets $\mbf{W}$, one can convert the data exchange problem
into a multicast network problem. For instance, considering the user $1$ as a receiver (see Figure~\ref{fig:netcode}),
it obtains the side information from himself (thus the link of capacity $\ell_1$ from user $1$ to user $1$),
and it receives transmissions from the other users through the links of capacities $R_i$, $i=2,3,\ldots,m$.
But, in order to set up the problem this way it is necessary to know
how many symbols in $\mathbb{F}_{q^n}$ each user broadcasts, \emph{i.e.}, we need to know the capacities $R_i$ of the links.

In \cite{ho2006random}, the authors proved necessary and sufficient conditions for the existence
of the network coding solution when the sources are linearly correlated. In the following Lemma we state their result
adapted to the data exchange problem with linearly coded packets.
Let us denote by $\mbf{A}_j(\mcl{S}_i,\star)$ a sub-matrix of $\mbf{A}_j$ with rows indexed by the elements of the set $\mcl{S}_i$.
\begin{lemma}\label{netcode:lm1}
In the data exchange problem with linearly coded packets, a rate tuple $(R_1,R_2,\ldots,R_m)$ can be achieved by network coding if and only if
\begin{align}
&{\text {rank}}(\mbf{A}_1, \mbf{A}_2(\mcl{S}^{(1)}_2,\star), \ldots, \mbf{A}_m(\mcl{S}^{(1)}_m,\star))=N, \label{netcode:eq1} \\
&{\text {rank}}(\mbf{A}_1(\mcl{S}^{(2)}_1,\star), \mbf{A}_2, \ldots, \mbf{A}_m(\mcl{S}^{(2)}_m,\star))=N, \label{netcode:eq2} \\
&~~~\vdots~~~~~~~~~~~~~~~~~~~~~~~~~~~~~~~~~~~~~~~~~~~~~\vdots  \nonumber \\
&{\text {rank}}(\mbf{A}_1(\mcl{S}^{(m)}_1,\star), \ldots, \mbf{A}_{m-1}(\mcl{S}^{(m)}_{m-1},\star),\mbf{A}_m)=N, \label{netcode:eq3}
\end{align}
such that $|\mcl{S}^{(j)}_i|=R_i$, $\forall j \in \{1,2,\ldots,m\}$, $\forall i \in \{1,2,\ldots,m\} \setminus \{j\}$,
where  $\mcl{S}^{(j)}_i\subseteq \{1,2,\ldots, \ell_i\}$.
\end{lemma}
Each equation in \eqref{netcode:eq1}-\eqref{netcode:eq3} corresponds to the selection of $N$ disjoint paths from
the users $1$ through $m$, to one of the receiving users (see Figure~\ref{fig:netcode} where user $1$ is the receiving user).
Hence, for a rate tuple $(R_1,R_2,\ldots,R_m)$ that satisfies the conditions in Lemma~\ref{netcode:lm1}
there exists a network coding solution to the data exchange problem.
Now, let us consider the equations \eqref{netcode:eq1} through \eqref{netcode:eq3}. The idea is to identify
the set of all achievable solutions for each receiver, \emph{i.e.}, the goal is to find the collection of sets $\{\mcl{S}_i^{(j)}\}_{i=1,i\neq j}^m$
for each $j \in \{1,2,\ldots,m\}$ which satisfy the conditions of the $j^{th}$ row in \eqref{netcode:eq1}-\eqref{netcode:eq3}.
To that end let us consider Algorithm~\ref{alg:greedy} (see \cite{schrijver2003combinatorial}).

\begin{algorithm}
\caption{Greedy Algorithm}
\label{alg:greedy}
\begin{algorithmic}[1]
\STATE Initialize $j(1) \in \{1,2,\ldots,m\}$, $\mbf{S}=\mbf{A}_{j(1)}$.
\STATE Let $j(2),j(3),\ldots,j(m)$ be any ordering of $\{1,2,\ldots,m\}\setminus \{j(1)\}$.
\FOR {$i=2$ to $m$}
\STATE Initialize $\mcl{S}^{(j(1))}_{j(i)}=\emptyset$.
\FOR {$k=1$ to $\ell_{j(i)}$}
\IF{
$\text{rank}(\mbf{S},\mbf{A}_{j(i)}(k,\star)) = \text{rank}\{\mbf{S}\} + \text{rank}\{\mbf{A}_{j(i)}(k,\star)\}$
}
\STATE
$\mbf{S}=\left[\begin{array}{c}
                            \mbf{S} \\
                            \mbf{A}_{j(i)}(k,\star) \\
                          \end{array}
                        \right],~~~
\mcl{S}^{(j(1))}_{j(i)}=\mcl{S}^{(j(1))}_{j(i)} \cup \{k\}$.
\ENDIF
\ENDFOR
\ENDFOR
\end{algorithmic}
\end{algorithm}
\newpage
It is not hard to conclude that Algorithm~\ref{alg:greedy} satisfies the maximum rank property, \emph{i.e.},
for every $j(1)\in\{1,2,\ldots,m\}$ it holds that
\begin{align}
&\text{rank}(\mbf{A}_{j(1)},\mbf{A}_{j(2)}(\mcl{S}^{(j(1))}_{j(2)},\star),\ldots,\mbf{A}_{j(i)}(\mcl{S}^{(j(1))}_{j(i)},\star))
\nonumber \\
&= \text{rank}(\mbf{A}_{j(1)},\mbf{A}_{j(2)},\ldots,\mbf{A}_{j(i)}),~~i=2,3,\ldots,m
\end{align}
Therefore, for one particular ordering $j(1),j(2),\ldots,j(m)$ of $1,2,\ldots,m$, we have that
\begin{align}
R_{j(i)} = \text{rank}(\mbf{A}_{j(1)},\mbf{A}_{j(2)},\ldots,\mbf{A}_{j(i)})-\text{rank}(\mbf{A}_{j(1)},\mbf{A}_{j(2)},\ldots,\mbf{A}_{j(i-1)}),
           ~~i=2,3,\ldots,m. \label{netcode:eq4}
\end{align}
From \eqref{netcode:eq4} it follows that
\begin{align}
\sum_{i=t}^m R_{j(i)} &= \text{rank}(\mbf{A}_{j(1)},\mbf{A}_{j(2)},\ldots,\mbf{A}_{j(m)})
                        -\text{rank}(\mbf{A}_{j(1)},\mbf{A}_{j(2)},\ldots,\mbf{A}_{j(t-1)}) \nonumber \\
&= \text{rank}(\mbf{A}_{j(t)},\mbf{A}_{j(t+1)},\ldots,\mbf{A}_{j(m)}|\mbf{A}_{j(1)},\mbf{A}_{j(2)},\ldots,\mbf{A}_{j(t-1)}),
~~t=2,3,\ldots,m.
\end{align}
Since the feasibility condition has to be satisfied for any ordering, we conclude that if
for every ordering \\ $j(1),j(2),\ldots, j(m)$ of $1,2,\ldots,m$
\begin{align}
\sum_{i=t}^m R_{j(i)}
\geq \text{rank}(\mbf{A}_{j(t)},\mbf{A}_{j(t+1)},\ldots,\mbf{A}_{j(m)}|\mbf{A}_{j(1)},\mbf{A}_{j(2)},\ldots,\mbf{A}_{j(t-1)}),
~~t=2,3,\ldots,m, \label{netcode:cond}
\end{align}
then $(R_1,R_2,\ldots,R_m)$ can be achieved by network coding.
It is not hard to see that the rate region in \eqref{netcode:cond} is equivalent to
\begin{align}
\sum_{i \in \mcl{S}} R_i \geq \text{rank}(\mbf{A}_{\mcl{S}}|\mbf{A}_{\mcl{S}^c}),~~\forall \mcl{S}\subset \{1,2,\ldots,m\}. \label{netcode:cut}
\end{align}
Thus, we showed that the cut-set bounds \eqref{netcode:cut} for the data exchange problem with linearly coded packets can be achieved via network coding.


\end{document}